\documentclass[journal,onecolumn]{IEEEtran}

\usepackage{amsmath}
\usepackage{amsfonts}
\usepackage{float}
\usepackage{graphicx}
\usepackage{amsthm}
\usepackage{amssymb}
\newtheorem{assumption}{Assumption}
\newtheorem{problem}{Problem}
\newtheorem{theorem}{Theorem}
\newtheorem{algorithm}{Algorithm}
\newtheorem{lemma}{Lemma}
\newtheorem{remark}{Remark}
\newtheorem{corollary}{Corollary}
\usepackage{accents}
\usepackage{color}
\usepackage{algorithm}
\usepackage{algpseudocode}
\def\BState{\State\hskip-\ALG@thistlm}
\usepackage{supertabular,booktabs}
\usepackage{cite}
\usepackage{multicol,lipsum}
\usepackage{stfloats}
\usepackage{arydshln}
\usepackage{mathtools}
\usepackage{graphicx}
\usepackage{subcaption}
\usepackage[pagebackref=false,colorlinks=true, citecolor=red,linkcolor=cyan,urlcolor=blue,bookmarks=true]{hyperref}
\usepackage[hyperpageref]{backref}
\graphicspath{ {images/} }

\usepackage{tikz}
\usetikzlibrary{shapes,arrows}
\usetikzlibrary{positioning}	
	
\newtheoremstyle{theoremdd}
{\topsep}
{\topsep}
{\itshape}
{0pt}
{\bfseries}
{:}
{ }
{\thmname{#1}\thmnumber{ #2} \thmnote{(#3)}}

\theoremstyle{theoremdd}
\newtheorem{definition}{Definition}

\usepackage{chngcntr}
\usepackage{apptools}
\AtAppendix{\counterwithin{assumption}{section}}
\AtAppendix{\counterwithin{definition}{section}}
\AtAppendix{\counterwithin{theorem}{section}}

\newcommand{\RNum}[1]{\uppercase\expandafter{\romannumeral #1\relax}}
\newcommand{\sgmin}{\underline{\sigma}}
\newcommand{\sgmax}{\bar{\sigma}}
\newcommand\norm[1]{\lVert#1\rVert}
\DeclareMathOperator{\Tr}{Tr}
\DeclareMathOperator{\vecs}{vecs}
\DeclareMathOperator{\vecv}{vecv}
\DeclareMathOperator{\vect}{vec}

\begin{document}

\title{Robust Reinforcement Learning for Risk-Sensitive Linear Quadratic Gaussian Control}

\author{Leilei Cui, 
         Tamer Başar, ~\IEEEmembership{Life Fellow,~IEEE,} %
         and Zhong-Ping Jiang,~\IEEEmembership{Fellow,~IEEE}
\thanks{*This work has been supported in part by the NSF grants EPCN-1903781 and ECCS-2210320.}
\thanks{L. Cui and Z.~P.~Jiang are with the Control and Networks Lab, Department of Electrical and Computer Engineering, Tandon School of Engineering, New York University, Brooklyn, NY 11201, USA (e-mail: l.cui@nyu.edu;  zjiang@nyu.edu).}
\thanks{T. Başar is with the Coordinated Science Laboratory, University of Illinois Urbana-Champaign, Urbana, IL 61801 USA (e-mail: basar1@illinois.edu).}
}


\maketitle
\begin{abstract}
This paper proposes a novel robust reinforcement learning framework for discrete-time linear systems with model mismatch that may arise from the sim-to-real gap. A key strategy is to invoke advanced techniques from control theory. Using the formulation of the classical risk-sensitive linear quadratic Gaussian control, a dual-loop policy optimization algorithm is proposed to generate a robust optimal controller. The dual-loop policy optimization algorithm is shown to be globally and uniformly convergent, and robust against disturbances during the learning process. This robustness property is called small-disturbance input-to-state stability and guarantees that the proposed policy optimization algorithm converges to a small neighborhood of the optimal controller as long as the disturbance at each learning step is relatively small. In addition, when the system dynamics is unknown, a novel model-free off-policy policy optimization algorithm is proposed. Finally, numerical examples are provided to illustrate the proposed algorithm.
\end{abstract}
\begin{IEEEkeywords}
Robust reinforcement learning, policy optimization, risk-sensitive LQG.
\end{IEEEkeywords}

\section{Introduction}

By optimizing a specified accumulated performance index, reinforcement learning (RL), as a branch of machine learning, is aimed at learning optimal decisions from data in the absence of model knowledge. Policy optimization (PO) plays a pivotal role in the development of RL algorithms \cite[Chapter 13]{book_sutton}. The key idea of PO is to parameterize the policy and update the policy parameters along the gradient ascent direction of the performance index for maximization, or gradient descent for minimization. Since the system model is unknown, the policy gradient should be estimated by data-driven methods through sampling and experimentation. Consequently, accurate policy gradient can hardly be obtained because of various errors that may be induced by function approximation, measurement noise, and external disturbance. Therefore, both convergence and robustness properties of PO should be theoretically studied in the presence of gradient estimation error. 

Since linear quadratic regulator (LQR) is theoretically tractable and widely applied in various engineering fields, it stands out as providing an ideal benchmark for studying RL problems. For the LQR problem, the control policy is parameterized as a linear function of the state, i.e. $u_t = -Kx_t$. The corresponding performance index is $\mathcal{J}_{LQR}(K) = \sum_{t=0}^{\infty} \mathbb{E}( x_t^T Q x_t + u_t^T R u_t)$. PO for the LQR problem aims at solving the constrained optimization problem $\min_{K \in \mathcal{W}} \mathcal{J}_{LQR}(K)$, where $\mathcal{W}$ is the admissible set of stabilizing control policies. Since $\mathcal{J}_{LQR}(K)$ can be expressed in terms of a Lyapunov equation, which depends on $K$, $\mathcal{J}_{LQR}(K)$ is differentiable in $K$. Based on this result, standard gradient descent, natural policy gradient, and Newton gradient algorithms have been developed to minimize the performance index $\mathcal{J}_{LQR}(K)$ \cite{Bu2020_lqr,Bu2019_lqrdiscrete,Mohammadi2022,Fazel_2018,Gravell2021,Li2021}. Interestingly, the Newton gradient algorithm with a step size of $\frac{1}{2}$ is equivalent to the celebrated Kleinman's policy iteration algorithm \cite{Hewer1971,Kleinman1968, book_Jiang, Book_Lewis}. By the coercive property of the performance index $\mathcal{J}_{LQR}(K)$ ($\mathcal{J}_{LQR}(K) \to  \infty$ as $K \to \partial \mathcal{W}$), the stability of the updated control policy is maintained during PO. Furthermore, the global linear convergence rate of the algorithms is theoretically demonstrated by the gradient dominance property which is shown in \cite[Remark 2]{Mohammadi2022} and \cite[Lemma 3]{Fazel_2018}.  One of the reasons for using PO in these model-based approaches (perhaps the most important one) is that it provides a natural pathway to model-free analysis, where the RL techniques come into play. For example, when the system model is unknown, zeroth-order methods are applied to approximate the gradient of the performance index, supported by a sample complexity analysis \cite{Fazel_2018,Li2021,Mohammadi2022}. Since the estimation error of policy gradient is inevitable at each iteration, whether the error will accumulate and whether the adopted algorithms still converge in the presence of estimation error should be further studied.  By considering the PO algorithm as a nonlinear discrete-time system and invoking input-to-state stability \cite{Sontag2008}, the authors of \cite{Pang2021,Pang2022} show that the Kleinman's policy iteration algorithm can still find a near-optimal control policy even under the influence of estimation error. A similar robustness property is investigated for the steepest gradient descent algorithm \cite{Sontag2022}.

The aforementioned PO for the LQR problem cannot guarantee the robustness of the closed-loop system. For example, the obtained controller may fail to stabilize the system in the presence of model mismatch that may be induced by the sim-to-real gap and parameter variation. Risk-sensitive linear quadratic Gaussian (LQG) control was first proposed by \cite{Jacobson1973,Whittle1981}, which generalizes the risk-neutral optimal control (i.e. LQR) by minimizing the expectation of the accumulative quadratic cost transformed by the exponential function. It was shown in \cite{glover1988state} and \cite{book_Basar} that the risk-sensitive LQG control is equivalent to the mixed $\mathcal{H}_2$/$\mathcal{H}_\infty$ control and the linear quadratic zero-sum dynamic game. Therefore, it can guarantee the stability of the closed-loop system even under model mismatch. The authors of \cite{BORKAR2001339,Borkar2002,mihatsch2002risk} proposed RL algorithms for solving the model-free risk-sensitive control, but the methods are only applicable to Markov decision processes whose state and action spaces are finite. In \cite{Tamimi2007,Yufeng2021}, the authors proposed Q-learning algorithms for linear quadratic zero-sum dynamic games. The paper \cite{ZhangSICON} proposed PO methods for mixed $\mathcal{H}_2/\mathcal{H}_\infty$ control to guarantee robust stability of the closed-loop system. Through the concept of implicit regularization, it was shown that the proposed PO algorithms can find the globally optimal solution of mixed $\mathcal{H}_2/\mathcal{H}_\infty$ control at globally sublinear and locally superlinear rates. As there is a fundamental connection between mixed $\mathcal{H}_2/\mathcal{H}_\infty$ control and linear-quadratic zero-sum dynamic games (LQ ZSDGs), the natural policy gradient and Newton algorithms have been equivalently transformed into provably convergent dual-loop PO algorithms for the LQ ZSDG \cite{Zhang2019_game,Bu2019_game,ZhangNeuIps2020}. The outer loop is to learn a protagonist under the worst-case adversary while the inner loop is to learn a worst-case adversary. In this way, the protagonist can robustly perform the control tasks under the disturbances created by the adversary. Interestingly, the Newton algorithm with a step size of $\frac{1}{2}$ is equivalent to the policy iteration algorithm for ZSDG \cite{Lewis2007,Lewis_2006,Kyriakos2012}. In the aforementioned papers, the convergence of the learning algorithms for the risk-sensitive control has been analyzed under the ideal noise-free case. Besides the convergence, a useful learning algorithm should be robust and is capable of finding a near-optimal solution even in the face of noise that may be induced by noisy experimental data, rounding errors of numerical computation, or the finite-time stopping of the inner loop in the dual-loop learning setup. However, the issues of  uniform convergence and robustness of the dual-loop learning algorithm are still unsolved.

\subsection{Our Contributions in this Paper}
 A fundamental challenge of the convergence of the dual-loop PO algorithm is to address the uniform convergence issue tied to the inner loop. Specifically, the required number of inner-loop iterations should be independent of the outer-loop iteration. Otherwise, as the outer-loop iteration increases, the required number of inner-loop iterations may grow explosively, thus making the dual-loop algorithm not practically implementable. To the best of our knowledge, the uniform convergence of the dual-loop algorithm has not been theoretically analyzed heretofore.

In addition, PO algorithm cannot be implemented accurately in practical applications, due to the influence of various errors arising from gradient estimation error, sensor noise, external disturbance, and modeling error. Hence, a fundamental question arises: Is the  PO algorithm robust to the errors? In particular, does the PO algorithm still converge to a neighbourhood of the optimal solution in the presence of various errors and, if yes, what is the size of the neighborhood? For both the outer and inner loops, the iterative process is nonlinear, and the robustness of the PO algorithm has not been fully understood in the present literature. 

In this paper, we investigate uniform convergence and robustness of the dual-loop iterative algorithm for solving the problem of risk-sensitive linear quadratic Gaussian control. Even though the convergence of the dual-loop iterative algorithm is analyzed separately in \cite{ZhangSICON,ZhangNeuIps2020,Fazel_2018}, uniform convergence and robustness of the overall algorithm are still open problems. To analyze the uniform convergence of the dual-loop algorithm, the key idea is to demonstrate global linear convergence of the inner-loop iteration and find an upperbound on the linear convergence rate. To address the robustness issue, a key strategy of the paper is to invoke techniques from advanced control theory, such as input-to-state stability (ISS) \cite{Sontag2008} and its latest variant called ``small-disturbance ISS'' \cite{Pang2021} to analyze the robustness of the proposed discrete-time iterative algorithm. In the presence of noise during the learning process, it is demonstrated that the PO algorithm still converges to a small neighbourhood of the optimal solution, as long as the noise is relatively  small. Furthermore, based on these results and the technique of approximate dynamic programming \cite{tutorial_Jiang,book_Bertsekas_Neu}, an off-policy data-driven RL algorithm is proposed when the system is disturbed by an immeasurable Gaussian noise. Several numerical examples are given to validate the efficacy of our theoretical results. 

To sum up, our main contributions in this paper are three-fold: 1) the uniform convergence of the dual-loop iterative algorithm is theoretically analyzed; 2) under the framework of the small-disturbance ISS, the robustness of both the outer and inner loops is theoretically demonstrated; 3) a novel learning-based off-policy policy optimization algorithm is proposed.

A shorter and preliminary version of this paper was presented at the conference L4DC 2023 \cite{CuiL4DC2023}. Compared with the conference paper, in this paper, we provide rigorous proofs for all the theoretical results. In addition, in Section V-A, we propose a method to learn an initial admissible controller. Finally, a benchmark example known as cart-pole system is provided to demonstrate the effectiveness of the proposed dual-loop algorithm.

\subsection{Organization of the Paper}
Following this Introduction section, Section \RNum{2} provides some preliminaries on linear exponential quadratic Gaussian (LEQG) control problem, LQG zero-sum dynamic game, and robustness analysis. Section \RNum{3} introduces the model-based dual-loop iterative algorithm to optimize the policy for LEQG control, and the convergence of the algorithm is analyzed. Section \RNum{4} analyzes robustness of the dual-loop iterative algorithm to various errors in the learning process within the framework of ISS. Section \RNum{5} presents a learning-based policy optimization algorithm for LEQG control. Section \RNum{6} provides two numerical examples to illustrate the proposed algorithm. The paper ends with the concluding remarks of Section \RNum{7} and six appendices which include proofs of the main results in the main body of the paper.

\subsection{Notations}
$\mathbb{R}$ and $\mathbb{C}$ are the sets of real and complex numbers, respectively. $\mathbb{Z}$ ($\mathbb{Z}_+$) is the set of (positive) integers. $\mathbb{S}^{n}$ is the set of $n$-dimensional real symmetric matrices. $|a|$ denotes the Euclidean norm of a vector $a$. $\norm{\cdot}$ and $\norm{\cdot}_F$ denote the spectral norm and the Frobenius norm of a matrix. $\ell_2$ is the space of square-summable sequences equipped with the norm $\norm{\cdot}_{2}$. $\ell_\infty$ is the space of bounded sequences equipped with the norm $\norm{\cdot}_\infty$. $\sgmax(\cdot)$ and $\sgmin(\cdot)$ are respectively the maximum and minimum singular values of a given matrix.  For a transfer function $G(z)$, its $\mathcal{H}_\infty$ norm is defined as $\norm{G}_{\mathcal{H}_\infty} := \sup_{\omega \in [0,2\pi]}\sgmax(G(e^{j\omega}))$, which is equivalent to $\norm{G}_{\mathcal{H}_\infty} := \sup_{ u \in \ell_2}\frac{\norm{Gu}_{2}}{\norm{u}_{_{2}}}$.

For a matrix $X \in \mathbb{R}^{m \times n}$, $\vect(X) := [x_1^T,\cdots,x_n^T]^T$, where $x_i$ is the $i$th column of $X$. For a matrix $P \in \mathbb{S}^n$, $\vecs(P) := [p_{1,1}, p_{1,2},\cdots,p_{1,n}, p_{2,2}, p_{2,3},\cdots,p_{n,n}]^T$, where $p_{i,j}$ is the $i$th row and $j$th column entry of the matrix $P$. $[X]_{i}$ dentotes the $i$th row of $X$. $[X]_{i,j}$ denotes the submatrix of the matrix $X$ that is comprised of the rows between the $i$th and $j$th rows of $X$.  For a vector $a \in \mathbb{R}^{n}$, $\vecv(a) := [a_1^2, 2a_1a_2,\cdots,2a_1a_n,a_2^2,2a_2a_3,\cdots,a_n^2]^T$. $I_n$ denotes the $n$-dimensional identity matrix.

\section{Preliminaries}
\label{sec:background}

In this section, we begin with the problem formulation
of LEQG control.
Then, we discuss its relation to linear quadratic zero-sum dynamic games (DG) and its robustness analysis. 

\subsection{Linear Exponential Quadratic Guassian Control}
Consider the discrete-time linear time-invariant system
\begin{subequations}\label{eq:LTI}
\begin{align}
    x_{t+1} &= Ax_{t} + Bu_t +Dw_t \quad x_0 \sim \mathcal{N}(0,I_n),\label{eq:LTIsystem}\\
    y_t &= Cx_t + Eu_t, \label{eq:output}
\end{align} 
\end{subequations}%
where $x_t \in \mathbb{R}^n$ is the state of the system; $u_t \in \mathbb{R}^m$ is the control input; $x_{0}$ is the initial state; $w_t \in \mathbb{R}^q \sim \mathcal{N}(0,I_q)$ is independent and identically distributed random variable; $y_t \in \mathbb{R}^{p}$ is the controlled output. $A, \, B, \, C, \, D, \, E$ are constant matrices with compatible dimensions. 
The LEQG control problem entails finding an input sequence $u := \{u_t\}_{t=0}^{\infty}$, depending on the current value of the state, that is $\{u_t = \mu_t(x_t)\}^\infty_{t=0}$ where $\mu:= \{\mu_t:\mathbb{R}^n \to \mathbb{R}^m\}^\infty_{t=0}$ is a seqence of appropriately defined measurable control policies, such that the following risk-averse exponential quadratic cost is minimized
\begin{align}\label{eq:exponentCost}
    \mathcal{J}_{LEQG}(\mu) := \lim_{\tau \to \infty}\frac{2 \gamma^2}{\tau }\log \left[ \mathbb{E} \exp \left( \frac{1}{2 \gamma^2} \sum_{t=0}^{\tau} y^T_t y_t \right) \right]
\end{align}
where $\gamma$ is a positive constant. For proper formulation of the optimization problem, the following two assumptions are standard.

\begin{assumption}\label{ass:stabilizable}
$(A,B)$ is stabilizable, $C^TC = Q \succ 0$, and $
\gamma > \gamma_\infty$, where $\gamma_\infty > 0$ is the minimal value of $\gamma$ such that for $
\gamma > \gamma_\infty$ the solution to \eqref{eq:GARE} given below exists, or equivalently there exists a control under which \eqref{eq:exponentCost} is finite.
\end{assumption}
\begin{assumption}\label{ass:crossterm}
The matrices in \eqref{eq:output} satisfy $E^TE = R \succ 0$, and $C^T E = 0$.
\end{assumption}
Assumption \ref{ass:stabilizable} ensures the existence of a stabilizing solution to the LEQG control problem. As demonstrated in \cite[Theorem 3.8]{book_Basar}, $\gamma_\infty$ is finite. In addition, the positive definiteness of $Q \succ 0$ can be relaxed to $Q \succeq 0$, as long as $(A,C)$ is taken to be detectable. Assumption \ref{ass:crossterm} has two parts. The first, positive definiteness of the weighting matrix on control, is standard even in LQR. The second one is also a standard condition to simplify the LEQG control problem by eliminating the cross term in the cost \eqref{eq:exponentCost} between the control input $u$ and state $x$. Stabilizability of the pair $(A,B)$ implies that there exists a feedback gain $K \in \mathbb{R}^{m \times n}$ such that the spectral radius $\rho (A-BK) <1$. Henceforth, a matrix is stable if its spectral radius is less than $1$, and $K$ is stabilizing if $A-BK$ is stable. A feedback gain $K$ is called admissible if it belongs to the admissible set $\mathcal{W}$ defined in \eqref{eq:feasibleSet}. Assumptions \ref{ass:stabilizable} and \ref{ass:crossterm} are used throughout the paper. 

As investigated by \cite{Jacobson1973} and \cite[Lemma 2.1]{zhangarxiv2019}, for any admissible linear control policy $\mu_t(x_t) = -Kx_t$, the cost admits the closed-form:
\begin{align}\label{eq:LEQGclosedForm}
    \mathcal{J}_{LEQG}(K) = - \gamma^2 \log \det (I_n - \gamma^{-2}P_KDD^T),
\end{align}
where the matrix $P_K = P_K^T \succ 0$ is the unique solution to
\begin{subequations}\label{eq:AREforKPrelimi}
\begin{align}
     &(A-BK)^T U_{K} (A-BK) - P_K + Q + K^T R K = 0, \label{eq:AREforK}\\
      &U_K := P_K + P_KD( \gamma^2 I_q - D^TP_K D)^{-1}D^TP_K. \label{eq:UKExpre}
\end{align}    
\end{subequations}
Furthermore, the LEQG problem admits a unique optimal controller $u_t^* = -K^* x_t$, where 
\begin{align}\label{eq:Kopt}
    K^* = (R + B^TU^*B)^{-1} B^T U^* A.
\end{align}
with $P^* = (P^*)^T \succ 0$ the unique solution to the generalized algebraic Riccati equation (GARE)
\begin{subequations} \label{eq:GARE}
\begin{align}
    &(A-BK^*)^T U^* (A-BK^*) - P^* + Q + (K^*)^TRK^* = 0,   \label{eq:Popt} \\
    &U^* = P^* + P^*D( \gamma^2 I_q - D^TP^* D)^{-1}D^TP^*. \label{eq:Uopt}
\end{align}
\end{subequations}

\subsection{Linear Quadratic Zero-Sum Dynamic Game}
The dynamic game can be mathematically formulated as 
\begin{align}\label{eq:LQZSDG}
    &\min_{\mu}\max_{\nu}\mathcal{J}_{DG}(\mu, \nu) := \mathbb{E}_{x_0}\left(\sum_{t=0}^{\infty}y_t^Ty_t - \gamma^{2}w_t^Tw_t \right), \nonumber \\
    &\text{subject to} \,\, \eqref{eq:LTI},
\end{align}
where $u:=\{u_t\}_{t=0}^{\infty}$ and $w:=\{w_t\}_{t=0}^{\infty}$ are the input sequences for the minimizer and the maximizer, respectively, generated by state-feedback policies $\mu := \{\mu_t\}^\infty_{t=0}$ and $\nu := \{\nu_t\}^\infty_{t=0}$. Note that here in \eqref{eq:LTIsystem}, $w$ is no longer a Gaussian random sequence, but a second control variable, at the disposal of the maximizer.

For any admissible controller $\mu_t(x_t) = -Kx_t$, and with $\gamma > \gamma_\infty$, the closed-form cost is
\begin{align}\label{eq:costDGforK}
    \mathcal{J}_{DG}(K,\nu^*(K)) = \max_{\nu}\mathcal{J}_{DG}(-Kx_t,\nu) = \Tr(P_K),
\end{align}
where $P_K$ is the solution of \eqref{eq:AREforKPrelimi}. 
From \cite[Equation 3.51]{book_Basar}, it follows that the optimizers for the minimax problem are $\mu_t^*(x_t) = -K^*x_t$ and $\nu_t^*(x_t) = L^*x_t$, where $K^*$ is defined in \eqref{eq:Kopt} and $L^*$ is given by
\begin{align}\label{eq:Lopt}
    & L^* = (\gamma^{2}I_q - D^TP^*D )^{-1} D^T P^* (A-BK^*). 
\end{align}
Furthermore, $(A-BK^*)$ is stable, $I_q - \gamma^{-2}D^T P^*D \succ 0$, and $(A-BK^*+DL^*)$ is stable.

Therefore, the minimizer of DG shares the same optimal controller as the optimizer in the LEQG problem. Also note that, $P_K$ is critical for determining the closed-from costs of $\mathcal{J}_{LEQG}(K)$ and $\mathcal{J}_{DG}(K,\nu^*)$.

\subsection{Robustness Analysis}
With $w$ taken as a deterministic input in \eqref{eq:LTI}, and taking any stabilizing feedback $\mu_t(x_t) = -Kx_t$, the discrete-time transfer function from $w$ to $y$ can be expressed as
\begin{align}
    \mathcal{T}(K) := (C-EK)[zI_n - (A-BK)]^{-1}D.
\end{align}
where $z \in \mathbb{C}$ is the $z$-transform variable. 

Now consider the depiction in Fig. \ref{fig:smallgain}, where $\Delta$ denotes the model mismatch that may be induced by the sim-to-real gap, and satisfying $\norm{\Delta}_{\mathcal{H}_\infty} \leq  \frac{1}{\gamma}$. Thanks to the small-gain theorem \cite{book_zhou,Jiang2018,Zames1966}, when subjected to model mismatch, the system remains stable as long as $\norm{\mathcal{T}(K)}_{\mathcal{H}_\infty} < \gamma$. Consequently, the controller $\mu_t(x_t) = -Kx_t$ is robust to the model mismatch $\Delta$ if $K$ lies within the admissible set $\mathcal{W}$ defined as
\begin{align} \label{eq:feasibleSet}
    \mathcal{W} := \{K \in \mathbb{R}^{m \times n}| \rho(A-BK)<1,  \norm{\mathcal{T}(K)}_{\mathcal{H}_\infty} < \gamma\}.
\end{align}
As investigated in \cite[Theorem 3.8]{book_Basar}, the LEQG control in \eqref{eq:Kopt} satisfies $K^* \in \mathcal{W}$, and therefore, it is optimal with respect to \eqref{eq:exponentCost} and robust to the model mismatch. This motivates us to pose the following problem.

\begin{problem}
Design a learning-based control algorithm such that near-optimal control gains, i.e. approximate values of $K^*$, can be learned from the input-state data.
\end{problem}

We will first introduce the model-based PO algorithm whose convergence and robustness properties are instrumental for the learning-based algorithm.

\begin{figure}
    \centering
    \includegraphics[width=0.7\linewidth]{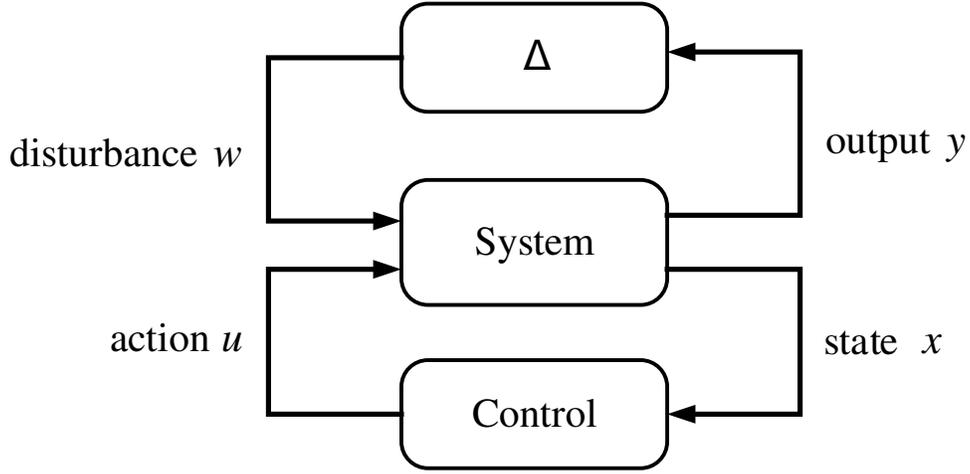}
    \caption{Robust control design with model mismatch $\Delta$.}
    \label{fig:smallgain}
\end{figure}

\section{Model-based Policy Optimization}
\label{sec:twoloop}
In this section, by resorting to the PO method,  a dual-loop iterative algorithm is proposed.

\subsection{Introduction of the Outer Loop}
The outer-loop iteration is developed based on the Newton PO algorithm in \cite{zhangarxiv2019}. For any $K \in \mathcal{W}$, the gradient $ \nabla_K\mathcal{J}(K) = \nabla_K \mathcal{J}_{LEQG}(K) = \nabla_K \mathcal{J}_{DG}(K,\nu^*(K))$ can be computed to be
\begin{align} \label{eq:GradientOuterloop}
\begin{split}
    &\nabla_K\mathcal{J}(K)= 2[(R+B^TU_KB)K - B^TU_KA] \Sigma_K,
\end{split}
\end{align}
where 
\begin{align}
    \Sigma_K &= \sum_{t=0}^{\infty}[(A-BK+DL_{K,*})^T]^tD \\
    &(I_q - \gamma^{-2}D^TP_KD)^{-1}D^T(A-BK+DL_{K,*})^t \nonumber. 
\end{align}
and 
\begin{align}\label{eq:LoptForK}
    L_{K,*} := ( \gamma^2 I_q - D^TP_K D)^{-1}D^TP_K(A-BK),
\end{align}
It is noticed that given $u_t=-Kx_t$, $L_{K,*}$ is considered as a worst-case feedback gain for $w$. \footnote{Henceforth, we write $\mu_t(x_t)=-Kx_t$ simply as $u_t = -Kx_t$, and likewise for $w$, as appropriate.} By the Newton method with a step size of $\frac{1}{2}$, to iteratively minimize $\mathcal{J}(K)$, the updated feedback gain is
\begin{align}\label{eq:GNPOupdate}
\begin{split}
    K' &= K - \frac{1}{2}(R+B^TU_KB)^{-1}\nabla_K \mathcal{J}(K) \Sigma_{K}^{-1} \\
    &= (R+B^TU_KB)^{-1}B^TU_KA.
\end{split}
\end{align}

Let $i$ denote the iteration index for the outer loop and introduce the following variables to simplify the notation
\begin{subequations}
\begin{align}
    & A_i := A-BK_i, \quad  Q_i := Q + K_i^T R K_i \label{eq:QKRK},
\end{align}
\end{subequations}
where $A_i$ is the closed-loop transition matrix with $u_t = -K_ix_t$, and $Q_i$ is the cost weighting matrix. Then, by \eqref{eq:AREforKPrelimi} and \eqref{eq:GNPOupdate}, the outer-loop iteration can be expressed as
\begin{subequations}\label{eq:outerloopIte}
\begin{align}
    & A_i^T U_{i} A_i - P_i + Q_i = 0, \label{eq:outloop_evalu}\\
    &    U_i := P_i + P_iD( \gamma^2 I_q - D^TP_i D)^{-1}D^TP_i, \label{eq:Ui}\\
    & K_{i+1} = (R + B^T U_i B)^{-1}B^T U_i A \label{eq:outerloop_update},
\end{align}
\end{subequations}
In \eqref{eq:outerloopIte}, from policy iteration perspective \cite{book_Bertsekas_Neu}, we consider \eqref{eq:outloop_evalu} as the policy evaluation step under the worst-case disturbance and \eqref{eq:outerloop_update} as the policy improvement step. $P_i$ is the cost matrix of \eqref{eq:LQZSDG} for the controller $u_t = -K_ix_t$ under the worst-case disturbance.

As seen in Lemma \ref{lm:outerloop_converge}, for each iteration, the controller $K_i$ generated by \eqref{eq:outerloopIte} preserves robustness to model mismatch, i.e. $K_i \in \mathcal{W}$. $P_i$ converges to $P^*$ with a globally sublinear and locally quadratic rate \cite[Theorems 4.3 and 4.4]{ZhangSICON}. We further investigate the convergence rate of \eqref{eq:outerloopIte} and rigorously demonstrate that $P_i$ monotonically converges to the optimal solution $P^*$ with a globally linear convergence rate. The proof of the following theorem can be found in Appendix B.

\begin{theorem}\label{thm:outerloopGloballinear}
Let $\mathcal{H} := \{\Tr(P_K)-\Tr(P^*)| K \in \mathcal{W} \}$. For any $h \in \mathcal{H}$ and $K_1 \in \mathcal{G}_h$, where $ \mathcal{G}_h = \{K \in \mathcal{W}| \Tr(P_K) \leq \Tr(P^*)+h \}$, there exists $\alpha(h) \in [0,1)$, such that 
\begin{align}
    \Tr(P_{i+1} - P^*) \leq \alpha(h) \Tr(P_{i} - P^*), \quad \forall i \in \mathbb{Z}_+.
\end{align}
\end{theorem}
Since $P_{i+1} - P^* \succeq 0$ and $ \norm{P_{i+1} - P^*}_F \leq \Tr(P_{i+1} - P^*) \leq \sqrt{n}\norm{P_{i+1} - P^*}_F$ (Lemma \ref{lm:normInequ}), the following inequlity holds
\begin{align}
   \norm{P_{i+1} - P^*}_F \leq \sqrt{n}\alpha^i(h) \norm{P_{1} - P^*}_F.
\end{align}
Since $\mathcal{J}_{DG}(K_i,\nu^*(K_i)) = \Tr(P_i)$ from \eqref{eq:costDGforK}, it follows that
\begin{align}
\begin{split}
    &\mathcal{J}_{DG}(K_{i+1},\nu^*(K_{i+1})) - \mathcal{J}_{DG}(\mu^*,\nu^*) \leq \\
    &\alpha(h) [\mathcal{J}_{DG}(K_i,\nu^*(K_i)) - \mathcal{J}_{DG}(\mu^*,\nu^*)].
\end{split}
\end{align}
Hence, the cost of the dynamic game (under the worst-case disturbance) converges to the saddle point at a linear convergence rate $\alpha(h) \in [0,1)$. 

In the following subsection, given $K_i$, by maximizing  $\mathcal{J}_{DG}(K_i, \nu)$ over $\nu$, the inner-loop iteration is developed to get the worst-case disturbance $\nu^*(K_i)$.

\subsection{Introduction of the Inner Loop}

 Given the feedback gain of the minimizer $K \in \mathcal{W}$, the inner loop iteratively finds the optimal controller for the maximizer $w^*(K)$ by solving \footnote{Here and below, we have used the control $w$ instead of the policy $\nu$, for simplicity of the notation.}
\begin{align}
    &\max_{w}\mathcal{J}_{DG}(K,w) = \sum_{t=0}^{\infty}y_t^Ty_t - \gamma^{2}w_t^Tw_t, \\
    &\text{subject to} \,\, x_{t+1} = (A-BK)x_t + Dw_t, x_0 \sim \mathcal{N}(0,I_n).   \nonumber
\end{align}
The optimal solution is $w^*_t(K) = L_{K,*}x_t$, where $L_{K,*}$ is defined in \eqref{eq:LoptForK}. For any admissible controller $w_t = Lx_t$ (with ($A-BK+DL$) stable), by \cite{Fazel_2018} the closed-form cost is
\begin{align}
    \mathcal{J}_{DG}(K,L) = \Tr(P_{K,L}),
\end{align}
where $P_{K,L}=P_{K,L}^T \succ 0$ is the solution to 
\begin{align}\label{eq:LyaforL}
\begin{split}
    &(A-BK+DL)^TP_{K,L}(A-BK+DL) \\
    &- P_{K,L} + Q + K^TRK - \gamma^2 L^TL = 0
\end{split}
\end{align}
The gradient of $ \nabla_L \mathcal{J}_{DG}(K,L)$ is 
\begin{align}
\begin{split}
     &\nabla_L \mathcal{J}_{DG}(K,L) = 2\left[(\gamma^2I_q - D^TP_{K,L}D)L \right. \\
     & \left. - D^TP_{K,L}(A-BK) \right]\Sigma_{K,L},   
\end{split}
\end{align}
where
\begin{align}
    \Sigma_{K,L} &= \sum_{t=0}^{\infty} (A-BK+DL)^t [(A-BK+DL)^T]^t.
\end{align}
By Newton methods, the updated feedback gain is
\begin{align}\label{eq:GNInner}
    L' &= L - \frac{1}{2}(\gamma^2I_q - D^TP_{K,L}D)^{-1} \nabla \mathcal{J}_{DG}(K,L) \Sigma_{K,L}^{-1} \nonumber \\
    &= (\gamma^2I_q - D^TP_{K,L}D)^{-1} D^T P_{K,L} (A-BK). 
\end{align}

Let $j$ denote the iteration index of the inner loop, and introduce the following variables to simplify the notation:
\begin{subequations}
\begin{align}
    A_{i,j} &:= A - BK_i + DL_{i,j},  Q_i := Q + K_i^TRK_i, \label{eq:Aij}\\
    A_{i,*} &:= A-BK_i +DL_{i,*}, A^* := A-BK^*+DL^*.\label{eq:Aopt}
\end{align}
\end{subequations}
By \eqref{eq:LyaforL} and \eqref{eq:GNInner}, the inner loop is designed as
\begin{subequations}\label{eq:innerloop}
\begin{align}
    &A_{i,j} ^TP_{i,j}A_{i,j} - P_{i,j} + Q_i - \gamma^2L_{i,j}^TL_{i,j} = 0, \label{eq:innerloop_eval}\\
    &L_{i,j+1} = (\gamma^2 I_q - D^TP_{i,j}D)^{-1}D^TP_{i,j}A_i.  \label{eq:innerloop_update}
\end{align}
\end{subequations}
From the policy iteration perspective, we consider \eqref{eq:innerloop_eval} as the policy evaluation step and \eqref{eq:innerloop_update} as the policy improvement step for the inner loop. $P_{i,j}$ is the cost matrix of $\mathcal{J}_{DG}(K_i,L_{i,j})$ with the state-feedback policies $\mu_t(x_t) = -K_{i}x_t$ and $\nu_t(x_t) = L_{i,j}x_t$. The inner-loop policy iteration possesses the monotonicity property and preserves stability, that is the sequence $\{P_{i,j}\}_{j=1}^{\infty}$ is monotonically increasing and upper bounded by $P_i$, and $A-BK_i+DL_{i,j}$ is stable. These results are stated in Lemma \ref{lm:innerloop convergence}. We prove that the inner loop globally and linearly converges to the optimal solution $P_{i}$. The details of the proof are given in Appendix C. This brings us to the following theorem, whose proof is in Appendix C.

\begin{theorem}\label{thm:innerloop_globallinear}
Given $L_{i,1} = 0$, for any $K_i \in \mathcal{W}$, there exists a constant $\beta(K_i) \in [0,1)$, such that 
\begin{align}
    \Tr(P_i - P_{i,j+1}) \leq \beta(K_i) \Tr(P_i - P_{i,j}), \quad \forall j \in \mathbb{Z}_+. 
\end{align}
\end{theorem}

Based on Theorem \ref{thm:innerloop_globallinear}, we can further obtain that 
\begin{align}
    \norm{P_{i} - P_{i,j+1}}_F \leq \sqrt{n}\beta^{j}(K_i)\norm{P_i - P_{i,1}}_F.
\end{align}
In addition, since $\mathcal{J}_{DG}(K_i,L_{i,j}) = \Tr(P_{i,j})$ and $\mathcal{J}_{DG}(K_i,L_{i,*}) = \Tr(P_{i})$, from Theorem \ref{thm:innerloop_globallinear}, we have 
\begin{align}
\begin{split}
    &\mathcal{J}_{DG}(K_i,L_{i,*}) - \mathcal{J}_{DG}(K_i,L_{i,j+1}) \leq \\
    &\beta(K_i)[\mathcal{J}_{DG}(K_i,L_{i,*}) - \mathcal{J}_{DG}(K_i,L_{i,j})].
\end{split}
\end{align}
Hence, the cost of the dynamic game with $K_i$ is monotonically increasing and converges to the maximum at a linear rate $\beta(K_i) \in [0,1)$. 

The dual-loop policy iteration algorithm is in Algorithm \ref{alg:IteAlg_model}. Ideally, $P_{i,j}$ generated by the inner loop converges to $P_i$, and then the control gain for the minimizer is updated to $K_{i+1}$ by \eqref{eq:outerloop_update}. In practice, the inner loop stops after $\bar{j}$ iterations, and $P_{i,\bar{j}}$, instead of $P_i$, is used for updating the control gain at line \ref{alg:line::outer_loop_update}. $K_{i,\bar{j}}$ denotes the control gain updated at the outer loop.

\subsection{Convergence Analysis for the Dual-Loop Algorithm}
For the dual-loop algorithm, the inner-loop iteration linearly converges to the optimal solution $P_K$ with the rate dependent on $K$. Since $K$ is updated iteratively, it is required that the inner loop enter the given neighborhood of $P_K$ within a constant number of steps, regardless of $K$. Otherwise, as the outer-loop iteration proceeds, the required number of inner-loop iterations may grow explosively, thus making the dual-loop algorithm not practically implementable. The uniform convergence rate of the overall algorithm is given in the following theorem, whose proof is given in Appendix D.  

\begin{theorem} \label{thm:uniformConvergence}
For any $h \in \mathcal{H}$, $K\in \mathcal{G}_h$, and $\epsilon>0$, there exists $\bar{j}(h,\epsilon) \in \mathbb{Z}_{+}$ independent of $K$, such that for all $j \geq \bar{j}(h,\epsilon)$, $\norm{P_{K,j} - P_{K}}_F \leq \epsilon$.
\end{theorem}

\begin{algorithm}[t] 
	\caption{Model-Based Policy Optimization}\label{alg:IteAlg_model}
    \begin{algorithmic}[1]
	\State Set $\bar{i},\bar{j} \in \mathbb{Z}_+$. Initialize $K_{1,\bar{j}} \in \mathcal{W}$ \;
	\For{$i \leq \bar{i}$}
		\State Initialize $j=1$ and $L_{i,1} = 0$\;		
		\State $Q_{i,\bar{j}} = C^TC + K_{i,\bar{j}}^T R K_{i,\bar{j}}$\;
		\For{$j \leq \bar{j}$}\label{alg:line::inner_loop_start}
		    \State $A_{i,j} = A - BK_{i,\bar{j}} + DL_{i,j}$\;
		    \State Get $P_{i,j}$ by solving
			\eqref{eq:innerloop_eval}\;
			\State Update $L_{i,j+1}$ by \eqref{eq:innerloop_update}\;
			\State $j \leftarrow j+1$ \;
        \EndFor
		\label{alg:line::inner_loop_end}
		\State $U_{i,\bar{j}} = P_{i,\bar{j}} + P_{i,\bar{j}}D(\gamma^{2}I_q-D^TP_{i,\bar{j}}D)^{-1}D^TP_{i,\bar{j}}$\;
		\State $K_{i+1,\bar{j}} = (R + B^T U_{i,\bar{j}} B)^{-1}B^T U_{i,\bar{j}} A$ \label{alg:line::outer_loop_update}\;
        \State $i \leftarrow i+1$ \;
    \EndFor
\end{algorithmic}
\end{algorithm}

\section{Robustness Analysis for the Dual-Loop Algorithm}
In the previous section, the exact PO algorithm was introduced in the sense that an accurate knowledge of system matrices was required to implement the algorithm. In practice, however, we do not have access to such an accurate model. Therefore, Algorithm \ref{alg:IteAlg_model} has to be implemented in a model-free setting. For example, we can approximate the gradient by zeroth-order method or approximate the value function (parameterized by cost matrix $P_K$) by approximate dynamic programming.  The updates for the controllers in \eqref{eq:outerloop_update} and \eqref{eq:innerloop_update} are subjected to noise. In this section, using the well-known concept of input-to-state stability in control theory, we will analyze the robustness of the dual-loop policy iteration algorithm in the presence of disturbance.

\subsection{Notions of Input-to-State Stability}
Consider the general nonlinear discrete-time system
\begin{align}\label{eq:generalsys}
    \chi_{k+1} = f(\chi_k, \rho_k).
\end{align}
where $\chi_k \in \mathcal{X}$, $\rho_k \in \mathcal{V}$, and $f$ is continuous. $\chi_e$ is the equilibrium state of the unforced system, that is $0 = f(\chi_e,0)$. 

\begin{definition}\cite{book_Hahn}
A function $\xi(\cdot): \mathbb{R}_+ \to \mathbb{R}_+$ is a $\mathcal{K}$-function if it is continuous, strictly increasing and vanishes at zero. A function $\kappa(\cdot,\cdot): \mathbb{R}_+ \times \mathbb{R}_+ \to \mathbb{R}_+$ is a $\mathcal{KL}$-function if for any fixed $t \geq 0$, $\kappa(\cdot,t)$ is a $\mathcal{K}$-function, and for any $r \geq 0$, $\kappa(r,\cdot)$ is decreasing and $\kappa(r,t) \to 0$ as $t \to \infty$. 
\end{definition}

\begin{definition}\cite{Jiang2001}
System \eqref{eq:generalsys} is ISS if there exist a $\mathcal{KL}$-function $\kappa$ and a $\mathcal{K}$-function $\xi$ such that for each input $\rho \in \ell_\infty$ and initial state $\chi_1 \in \mathcal{X}$, the following holds
\begin{align}
    \norm{\chi_k - \chi_e} \leq \kappa(\norm{\chi_1 - \chi_e}, k) + \xi(\norm{\rho}_\infty).
\end{align}
for any $k\in\mathbb{Z}_+$.
\end{definition}

Generally speaking, input-to-state stability characterizes the influence of input $\rho$ to the evolution of state $\chi$. The deviation of the state $\chi$ to the equilibrium is bounded as long as the input $\rho$ is bounded. Furthermore, the influence of the initial deviation $\norm{\chi_1 - \chi_e}$ vanishes as time tends to infinity. 

\subsection{Robustness Analysis for the Outer Loop}
The exact outer loop iteration is shown in \eqref{eq:outerloopIte}, and in the presence of disturbance, it is modified as
\begin{subequations}\label{eq:outerloopInext}
\begin{align}
    & \hat{A}_i^T \hat{U}_{i} \hat{A}_i - \hat{P}_i + \hat{Q}_i = 0, \label{eq:outloopInext_evalu}\\
    & \hat{U}_{i} = \hat{P}_{i} + \hat{P}_{i}D(\gamma^2 I_q - D^T \hat{P}_{i} D)^{-1}D^T \hat{P}_{i}, \\
    & \hat{K}_{i+1} = (R + B^T \hat{U}_i B)^{-1}B^T \hat{U}_i A + \Delta K_{i+1}, \label{eq:outerloopInext_update}
\end{align}
\end{subequations}
where $\Delta K_i$ is the disturbance at the $i$th iteration, and ``hat" is used to distinguish the sequences generated by the exact \eqref{eq:outerloopIte} and inexact \eqref{eq:outerloopInext} outer-loop iterations. By considering \eqref{eq:outerloopInext} as a discrete-time nonlinear system with the state $\hat{P}_i$ and input $\Delta K_i$, it can be shown that \eqref{eq:outerloopInext} is inherently robust to $\Delta K$ in the sense of small-disturbance ISS \cite{Pang2021,Pang2022}.

\begin{theorem}\label{thm:outerISS}
Given $K_1 \in \mathcal{G}_h$, there exists a constant $d(h) > 0$, such that if $\norm{\Delta K}_\infty < d(h)$, \eqref{eq:outerloopInext} is small-disturbance ISS. That is, there exist a $\mathcal{KL}$-function $\kappa_{1}(\cdot,\cdot)$ and a $\mathcal{K}$-function $\xi_1(\cdot)$, such that
\begin{align}
    \norm{\hat{P}_{i} - P^*}_F \leq \kappa_{1}(\norm{\hat{P}_{1} - P^*}_F, i) + \xi_1(\norm{\Delta K}_\infty).
\end{align}
\end{theorem}
We note that $\Delta K = \{\Delta K_i \}_{i=1}^{\infty}$ is a sequence of disturbance signals, and its $\ell_\infty$-norm is $\norm{\Delta K}_\infty = \sup_{i \in \mathbb{Z}_+}\norm{\Delta K_i}_F$.

If the outer-loop update in \eqref{eq:outerloopIte} can be computed exactly, which requires $P_i$ to be exact, then the iterations of the exact update will not leave the admissible set $\mathcal{W}$. In contrast, the dual-loop algorithm (Algorithm \ref{alg:IteAlg_model}) has access only to $P_{i,\bar{j}}$, which is close to $P_i$ but not exact. There is the possibility that this inaccuracy could drive the outer-loop iteration away from the optimal solution or even beyond the admissible region $\mathcal{W}$. As a direct corollary to Theorems \ref{thm:uniformConvergence} and \ref{thm:outerISS}, we state below that Algorithm \ref{alg:IteAlg_model} can still find a near-optimal solution if $\Bar{i}$ and $\Bar{j}$ are large enough.

\begin{corollary}\label{cor:dual-loopaccu}
For any $h\in \mathcal{H}$, $K_1 \in \mathcal{G}_h$, and $\epsilon>0$, there exist  $\bar{i}(h,\epsilon) \in \mathbb{Z}_+$ and $\bar{j}(h,\epsilon) \in \mathbb{Z}_+$, such that $\norm{{P}_{\bar{i},\bar{j}} - {P}^*}_F < \epsilon$.
\end{corollary}

\subsection{Robustness Analysis for the Inner Loop}
As a counterpart of the inexact outer-loop iteration, the inexact inner-loop iteration can be developed as
\begin{subequations}\label{eq:innerloopInext}
\begin{align}
    &\hat{A}_{i,j} ^T\hat{P}_{i,j}\hat{A}_{i,j} - \hat{P}_{i,j} + \hat{Q}_i - \gamma^2\hat{L}_{i,j}^T\hat{L}_{i,j} = 0, \label{eq:innerloopInext_eval}\\
    &\hat{L}_{i,j+1} = (\gamma^2 I_q - D^T\hat{P}_{i,j}D)^{-1}D^T\hat{P}_{i,j}\hat{A}_i + \Delta L_{i,j+1}.  \label{eq:innerloopInext_update}
\end{align}
\end{subequations}
Here, $\Delta L_{i,j+1}$ denotes the disturbance to the inner loop iteration and ``hat" emphasizes that the corresponding sequences are generated by the inexact iteration. With the inexact inner loop at hand, the following theorem shows that the inner-loop iteration \eqref{eq:innerloopInext} is robust to disturbance $\Delta L_{i}$ in the sense of small-disturbance input-to-state stability \cite{Pang2021,Pang2022}.

\begin{theorem} \label{thm:innerISS}
For any $\hat{K}_i \in \mathcal{W}$, there exists a constant $e(\hat{K}_i) > 0$, such that if $\norm{\Delta L_i}_\infty < e(\hat{K}_i)$, \eqref{eq:innerloopInext} is small-disturbance ISS. That is, there exist a $\mathcal{KL}$-function $\kappa_{2}(\cdot,\cdot)$ and a $\mathcal{K}$-function $\xi_2(\cdot)$, such that
\begin{align}
    \norm{\hat{P}_{i,j} - \hat{P}_{i}}_F \leq \kappa_{2}(\norm{\hat{P}_{i,1} - \hat{P}_{i}}_F, j) + \xi_2(\norm{\Delta L_i}_\infty).
\end{align}
\end{theorem}
We note that $\Delta L_i = \{\Delta  L_{i,j} \}_{j=1}^{\infty}$ is a sequence of disturbance signals, and $\norm{\Delta  L_i }_\infty = \sup_{j \in \mathbb{Z}_+}\norm{\Delta L_{i,j}}_F$.
The results of these last two theorems guarantee robustness of the dual-loop PO algorithm. Literally speaking, when the dual-loop PO algorithm is implemented in the presence of disturbance, it still finds the near optimal solution, and the deviation between the generated policy and the optimal one is determined by the magnitude of the disturbance. To be more specific, as iteration $i$ ($j$ for inner loop) goes to infinite, the cost matrix $\hat{P}_{i}$ ($\hat{P}_{i,j}$) enters a small neighborhood of the optimal solution $P^*$ ($\hat{P}_i$).

\section{Learning-based Off-Policy Policy Optimization}
We will develop a learning-based algorithm to learn from data a robust suboptimal controller (i.e. an approximation of $K^*$) without requiring any accurate knowledge of $(A,B,D)$ under the setting of zero-sum dynamic game with additive Gaussian noise. The input-state data from the following system will be utilized for learning:
\begin{align}\label{eq:zeroSumnoise}
    x_{t+1} = Ax_t + Bu_t + D w_t + v_t,
\end{align}
where $v_t \sim \mathcal{N}(0,\Sigma)$ with $\Sigma=\Sigma^T \succeq 0$ is independent and identically distributed noise.

\subsection{Learning-Based Policy Optimization}
Suppose that the exploratory policies for the minimizer and maximizer are  
\begin{subequations}\label{eq:exploreinput}
\begin{align}
    u_t &= -\hat{K}_{exp} x_t + \sigma_1\xi_1,  &\xi_1 \sim \mathcal{N}(0,I_m), \\
    w_t &= \hat{L}_{exp} x_t + \sigma_2\xi_1,  &\xi_2 \sim \mathcal{N}(0,I_q),
\end{align}    
\end{subequations}
where $\hat{K}_{exp} \in \mathbb{R}^{m \times n}$ and $\hat{L}_{exp}\in \mathbb{R}^{q \times n}$ are exploratory feedback gains, and $\sigma_1,\sigma_2>0$ are the standard deviations of the exploratory noise. Let $z_t := [x_t^T, u_t^T, w_t^T]^T$. For any given matrices $X \in \mathbb{S}^{n}$, let
\begin{align}\label{eq:GammaDef}
\begin{split}
    &\Gamma(X) := \begin{bmatrix}
        \Gamma_{xx}(X) & \Gamma_{xu}(X) &\Gamma_{xw}(X) \\
        \Gamma_{ux}(X) & \Gamma_{uu}(X) &\Gamma_{uw}(X) \\
        \Gamma_{wx}(X) & \Gamma_{wu}(X) &\Gamma_{ww}(X) \\
    \end{bmatrix} \\
    &= \begin{bmatrix}
        A^T X A + Q & A^T X B &A^T X D \\
        B^T X A & B^T X B + R &B^T X D \\
        D^T X A & D^T X B &D^T X D - \gamma^2I_q
    \end{bmatrix}. 
\end{split}
\end{align}   
Along the trajectories of system \eqref{eq:zeroSumnoise}, $x_{t+1}^TXx_{t+1}$ can be computed to be
\begin{align}\label{eq:xXx}
\begin{split}
    x_{t+1}^TXx_{t+1} &= z_t^T \Gamma(X)z_t - r_t + v_t^TXv_t \\
    &+ 2v_t^TX(Ax_t + Bu_t + Dw_t).
\end{split}
\end{align}
where $r_t := x_t^TQx_t + u_t^T R u_t - \gamma^2 w_t^Tw_t$ is the stage cost of the zero-sum dynamic game. Taking the expectation of \eqref{eq:xXx} results in
\begin{align}\label{eq:xXxexpec}
    \mathbb{E}\left[{z}_t^T \Gamma(X) {z}_t + \Tr(\Sigma X) - r_t  - x_{t+1}^TXx_{t+1} | z_t \right] = 0.
\end{align}
To simplify the notations, let 
\begin{align}
    \bar{z}_t &:= [\vecv(z_t)^T, 1]^T\\
    \theta(X) &:= [\vecs(\Gamma(X))^T, \Tr(\Sigma X)]^T.
\end{align}
By pre-multiplying \eqref{eq:xXxexpec} with $\bar{z}_t$, one can obtain 
\begin{align}\label{eq:xXxexpec2} 
\begin{split}
    &\mathbb{E}\left[\bar{z}_t \bar{z}_t^T \theta(X) - \bar{z}_tr_t - \bar{z}_t\vecv(x_{t+1})^T\vecs(X) | z_t \right] = 0.
\end{split}
\end{align}
\begin{assumption}\label{ass:ergodic}
There exists an ergodic stationary probability measure $\pi$ on $\mathbb{R}^{n+m+q}$ for system \eqref{eq:zeroSumnoise} with controller \eqref{eq:exploreinput}.
\end{assumption}
\begin{remark}
    Assumption \ref{ass:ergodic} is widely used in approximate dynamic programming and in the RL literature \cite{Tsitsiklis1994,Tsitsiklis1997}.
\end{remark}
\begin{assumption}\label{ass:invertible}
$\mathbb{E}_\pi\left[ {z}_t{z}_t^T \right]$ and $\mathbb{E}_\pi\left[ \bar{z}_t\bar{z}_t^T \right]$ are invertible.
\end{assumption}
\begin{remark}
Assumption \ref{ass:invertible} is reminiscent of the persistent excitation (PE) condition \cite{Jiang_book2021,astrom1997}. As in the literature of data-driven control \cite{book_Jiang, Book_Lewis2013}, one can satisfy it by means of added exploration noise, such as sinusoidal signals or random noise.
\end{remark}%

Under Assumptions \ref{ass:ergodic} and \ref{ass:invertible}, taking the expectation of \eqref{eq:xXxexpec2} with respect to the invariant probability measure $\pi$, we have
 \begin{align}\label{eq:leastsquareTheta}
     \theta(X) &= \Phi^{\dagger}\Xi\vecs(X) + \Phi^{\dagger}\Psi.
\end{align}
where 
\begin{align}
\begin{split}
    \Phi &:= \mathbb{E}_\pi\left[\bar{z}_t\bar{z}_t^T\right], \quad \Xi := \mathbb{E}_\pi\left[\bar{z}_t \vecv(x_{t+1})^T \right], \\
    \Psi &:= \mathbb{E}_\pi\left[\bar{z}_t r_t  \right].
\end{split}
\end{align}
Therefore, $\Gamma_{ww}(X)$ and $\Gamma(X)$ can be reconstructed as
\begin{subequations}
\begin{align}
    \vecs(\Gamma_{ww}(X)) &= [\Phi^{\dagger}]_{n_1,n_2}\Xi\vecs(X) + [\Phi^{\dagger}]_{n_1,n_2}\Psi, \label{eq:Gmwwdata}\\
    \vecs(\Gamma(X)) &= [\Phi^{\dagger}]_{1,n_2}\Xi\vecs(X) + [\Phi^{\dagger}]_{1,n_2}\Psi, \label{eq:Gmdata}
\end{align}    
\end{subequations}
where 
\begin{align}
\begin{split}
    n_1 &= \frac{(n+m+q)(n+m+q+1)}{2}+1-\frac{(q+1)q}{2},  \\
    n_2 &= \frac{(n+m+q)(n+m+q+1)}{2}.
\end{split}
\end{align}

In practice, we use a finite number of trajectory samples to estimate $\Phi$, $\Xi$, and $\Psi$, that is 
\begin{align}\label{eq:dataMatrix}
\begin{split}
    &\hat{\Phi}_\tau := \frac{1}{\tau}\sum_{t=1}^{\tau}\bar{z}_t\bar{z}_t^T, \quad \hat{\Xi}_\tau := \frac{1}{\tau}\sum_{t=1}^{\tau}\bar{z}_t\vecv(x_{t+1})^T,   \\
    &\hat{\Psi}_\tau := \frac{1}{\tau}\sum_{t=1}^{\tau} \bar{z}_tr_t.
\end{split}
\end{align}
By the Birkhoff Ergodic Theorem \cite[Theorem 16.2]{book_Koralov}, the following relations hold \textit{almost surely}
\begin{subequations}
\begin{align}
    &\lim_{\tau \to \infty}\hat{\Phi}_\tau = \Phi, \quad \lim_{\tau \to \infty}\hat{\Xi}_\tau = \Xi, \quad \lim_{\tau \to \infty}\hat{\Psi}_\tau = \Psi.
\end{align}
\end{subequations}
Then, by \eqref{eq:Gmdata}, ${\Gamma}(X)$ is estimated by a data-driven approach as follows:
\begin{align}\label{eq:leastsquareapproximateGm}
    \vecs(\hat{\Gamma}(X)) =  [\hat{\Phi}_\tau^{\dagger}]_{1,n_2}\hat{\Xi}_\tau\vecs(X)  + [\hat{\Phi}_\tau^{\dagger}]_{1,n_2}\hat{\Psi}_\tau,
\end{align}

With the data-driven estimate $\hat{\Gamma}(X)$, we will transform model-based PO in Algorithm \ref{alg:IteAlg_model} to a learning-based algorithm. Considering \eqref{eq:GammaDef}, we can rewrite \eqref{eq:innerloop_eval} as
\begin{align}\label{eq:PijData}
    &[I_n, -K_i^T, L_{i,j}^T]\Gamma(P_{i,j})[I_n, -K_i^T, L_{i,j}^T]^T - P_{i,j}= 0.
\end{align}
Vectorizing \eqref{eq:PijData} and plugging \eqref{eq:leastsquareTheta} into \eqref{eq:PijData} result in
\begin{align}\label{eq:PijDataVec}
    & \left\{\left([I_n, -K_i^T, L_{i,j}^T] \otimes [I_n, -K_i^T, L_{i,j}^T]\right) T_{n+m+q} \right.\nonumber \\
    &\left. [\Phi^{\dagger}]_{1,n_1}\Xi - T_n \right\} \vecs(P_{i,j}) = \\
    &-\left([I_n, -K_i^T, L_{i,j}^T] \otimes [I_n, -K_i^T, L_{i,j}^T]\right) T_{n+m+q} [\Phi^{\dagger}]_{1,n_1}\Psi, \nonumber
\end{align}
where $T_n$ and $T_{n+m+q}$ are the duplication matrices defined in Lemma \ref{lm:ComDupMat}. One can view \eqref{eq:PijDataVec} as a linear equation with respect to $\vecs(P_{i,j})$. Hence, at each inner-loop iteration, $\vecs(P_{i,j})$ can be computed by solving \eqref{eq:PijDataVec}. With $P_{i,j}$, according to \eqref{eq:innerloop_update} and \eqref{eq:GammaDef}, the feedback gain of the maximizer can be updated as
\begin{align}\label{eq:Lupdate}
    L_{i,j+1} = -\Gamma_{ww}(P_{i,j})^{-1}(\Gamma_{wx}(P_{i,j}) - \Gamma_{wu}(P_{i,j})K_i).
\end{align}

Next, we will develop a data-driven approach for the outer-loop iteration. According to the expression of $\Gamma_{ww}(X)$ in \eqref{eq:GammaDef}, and the duplication matrices defined Lemma \ref{lm:ComDupMat}, we have
\begin{align}\label{eq:Gmww}
\begin{split}
    \vecs(\Gamma_{ww}(X)) &= T_q^\dagger (D^T \otimes D^T)T_n \vecs(X) \\
    &- \gamma^2T_q^\dagger \vect(I_q).
\end{split}
\end{align}
Since \eqref{eq:Gmwwdata} and \eqref{eq:Gmww} hold for any $X \in \mathbb{S}^n$, it follows that
\begin{align}\label{eq:dataDD}
\begin{split}
    T_q^\dagger (D^T \otimes D^T)T_n &= [\Phi^{\dagger}]_{n_1,n_2}\Xi.
\end{split}
\end{align}
Then, from \eqref{eq:dataDD}, we obtain
\begin{align}\label{eq:dataDD2}
\begin{split}
    &(T_q^TT_q)[\Phi^{\dagger}]_{n_1,n_2}\Xi (T_n^TT_n)^{-1} = T_q^T(D^T \otimes D^T) (T_n^\dagger)^T 
\end{split}
\end{align}
For any $Y \in \mathbb{S}^q$, let
\begin{align}
    \Omega(Y) = D Y D^T.
\end{align}
According to Lemma \ref{lm:ComDupMat} and \eqref{eq:dataDD2}, it holds
\begin{align}\label{eq:DYDdata}
\begin{split}
    &\vecs(\Omega(Y)) = T_n^\dagger (D \otimes D)T_q \vecs(Y) \\
    &\quad = (T_n^TT_n)^{-1} \Xi^T[\Phi^{\dagger}]_{n_1,n_2}^T(T_q^TT_q)\vecs(Y).
\end{split}
\end{align}
Now, $\Omega(Y)$ can be computed by \eqref{eq:DYDdata} without knowing $D$. Then, according to \eqref{eq:outerloop_update}, the feedback gain for the minimizer is updated as
\begin{align}\label{eq:Kupdate}
\begin{split}
    U_{i} &= P_i - P_i\Omega(\Gamma_{ww}(P_i)^{-1})P_i, \\
    K_{i+1} &= \Gamma_{uu}(U_{i})^{-1}\Gamma_{ux}(U_{i}).
\end{split}
\end{align}

The learning-based PO algorithm is given in the table labeled as Algorithm \ref{alg:IteAlg_learning}. It should be noticed that in Algorithm \ref{alg:IteAlg_learning}, the system matrices ($A$, $B$, and $D$) are not involved in computing $\vecs(P_{i,j})$, $L_{i,j+1}$ and $K_{i+1}$. In addition, the updated controller $K_i$ is not applied for the data collection. Therefore, it is a learning-based off-policy algorithm for policy optimization.

\subsection{Learning an Initial Admissible Controller}
In Algorithm \ref{alg:IteAlg_learning}, an initial admissible feedback gain is required to start the learning-based policy optimization algorithm. In this section, we will develop a data-driven method for learning such an initial feedback gain.

 Taking the expectation of \eqref{eq:zeroSumnoise} with respect to the invariant probability measure $\pi$ in Assumption \ref{ass:ergodic}, we have
\begin{align}\label{eq:systemID}
    \mathbb{E}_\pi \left[x_{t+1}^T - z_t^T[A,B,D]^T\right] = 0.
\end{align}
Pre-multiplying \eqref{eq:systemID} by $z_t$ and using Assumption \ref{ass:invertible} yield
\begin{align}
    [A,B,D]^T = (\Phi')^\dagger \Xi',
\end{align}
where 
\begin{align}
    \Phi' = \mathbb{E}_\pi \left[z_t z_t^T\right], \quad \Psi' = \mathbb{E}_\pi \left[z_t x_{t+1}\right].
\end{align}
In practice, we can utilize a finite number of trajectory samples to estimate $\Phi'$ and $\Xi'$, i.e.
\begin{align}
\begin{split}
    &\hat{\Phi}'_\tau := \frac{1}{\tau}\sum_{t=1}^{\tau}{z}_t{z}_t^T, \quad \hat{\Xi}'_\tau := \frac{1}{\tau}\sum_{t=1}^{\tau}{z}_t(x_{t+1})^T.
\end{split}
\end{align}
By the Birkhoff Ergodic Theorem \cite[Theorem 16.2]{book_Koralov}, the following relations hold \textit{almost surely}
\begin{subequations}
\begin{align}
    &\lim_{\tau \to \infty}\hat{\Phi}_\tau = \Phi, \quad \lim_{\tau \to \infty}\hat{\Xi}_\tau = \Xi.
\end{align}
\end{subequations}
As a result, we obtain the estimates of the system matrices
\begin{align}\label{eq:systemIDABD}
    [\hat{A}_\tau,\hat{B}_\tau,\hat{D}_\tau]^T = (\hat{\Phi}'_\tau)^\dagger \hat{\Xi}_\tau.
\end{align}
With the identified system matrices, the linear matrix inequalities (LMIs) in \eqref{eq:dataLMI} can be solved for the initial admissible controller design. The admissibility of the obtained initial controller is guaranteed by Theorem \ref{thm:initialController} below, which is proved in Appendix G. 

\begin{theorem}\label{thm:initialController}
    There exist $\epsilon > 0$, $\mu > 0$, and $\tau^*(\epsilon,\mu) > 0$, such that for any $\tau > \tau^*(\epsilon,\mu)$, the following LMIs 
    \begin{subequations}\label{eq:dataLMI}
    \begin{align} 
    \begin{bmatrix}
    -{W} &* &* &* \\
    0  &-\gamma ^2  I_q &* &* \\
    \hat{A}_\tau{W} - \hat{B}_\tau{V} &\hat{D}_\tau &-{W} &* \\
    C{W} - E{V} &0 &0 &-I_p
    \end{bmatrix} &\prec -\epsilon I, \label{eq:hinfLMISysID1}\\
    -\begin{bmatrix}
         I_n & * & * \\
        \mu W  &I_n  &* \\
        \mu V &0  &I_m
    \end{bmatrix} &\prec 0, \label{eq:hinfLMISysID2}
    \end{align}
    \end{subequations}
     have a solution and $K = {V}{W}^{-1} $ that belongs to $\mathcal{W}$ almost surely.
\end{theorem}

\begin{algorithm}[t] 
	\caption{Learning-based Policy Optimization}\label{alg:IteAlg_learning}
 \begin{algorithmic}[1]
    \State  Initialize $\hat{K}_1 \in \mathcal{W}$ \; 
    \State Set $\bar{i},\bar{j} \in \mathbb{Z}_+$. \;    
	\State Set the length of the sampled trajectory $\tau$, and the exploration variances $\sigma_1^2$ and $\sigma_2^2$ \;
	\State Collect data from \eqref{eq:zeroSumnoise} with exploratory input \eqref{eq:exploreinput} \;
	\State Construct $\hat{\Phi}_\tau$,  $\hat{\Xi}_\tau$, and 
    $\hat{\Psi}_\tau$ defined in \eqref{eq:dataMatrix} \;
        \State Get the expressions of $\hat{\Gamma}(X)$ and $\hat{\Omega}(X)$ by \eqref{eq:leastsquareapproximateGm} and \eqref{eq:DYDdata}. \;
	\For{$i \leq \bar{i}$}  
		\State Set $\hat{L}_{i,1} = 0$\;
		\For{$j \leq \bar{j}$}
		    \State Get $\hat{P}_{i,j}$ by solving \eqref{eq:PijDataVec}\; 
			\State Update $\hat{L}_{i,j+1}$ by \eqref{eq:Lupdate} \;
            \State $j \leftarrow j+1$ \;
		\EndFor \;
		\State Update $\hat{K}_{i+1}$ by \eqref{eq:Kupdate} \label{alg:line::Kupdate} \;
        \State $i \leftarrow i+1$ \;
	\EndFor \;  
\end{algorithmic}
\end{algorithm}
\section{Numerical Simulations}
\subsection{An Illustrative Example}
We apply Algorithms \ref{alg:IteAlg_model} and \ref{alg:IteAlg_learning} to the system studied in \cite{Zhang2021_NIPS}. The system matrices are 
\begin{footnotesize}
\begin{align}
    A = \begin{bmatrix} 1 & 0 &-5\\ -1 &1 &0\\ 0 & 0 & 1 \end{bmatrix}, B = \begin{bmatrix} 1 &-10 & 0\\ 0 &3 &1\\ -1& 0 &2 \end{bmatrix}, 
    D = \begin{bmatrix} 0.5 & 0 & 0\\0 & 0.2 & 0\\ 0 & 0 & 0.2\end{bmatrix}. \nonumber
\end{align}
\end{footnotesize}%
The matrices related to the output are $C=[I_3,0_{3\times 3}]^T$ and $E=[0_{3\times 3}, I_3]^T$. The $\mathcal{H}_\infty$ norm threshold is $\gamma = 5$. $\bar{i} = 10$ and $\bar{j}=20$. 

The robustness of Algorithm \ref{alg:IteAlg_model} in the presence of disturbance is validated first. For each outer and inner loop iteration, the entries of the disturbances $\Delta K_i$ and $\Delta L_{i,j}$ are taken as samples from a standard Gaussian distribution and then their Frobenius norms are normalized to $0.09$. The algorithm is run independently for $50$ times. The relative errors of the gain matrix $\hat{K}_i$ and cost matrix $\hat{P}_{i}$, and the $\mathcal{H}_\infty$ norm of the closed-loop system with $\hat{K}_i$ are shown in Fig. \ref{fig:Norm_PI}, where the bold line is the mean of the trials and the shaded region denotes the variance. It is seen that with the disturbances at both the outer and inner loop iterations, the generated controller and the corresponding cost matrix approach the optimal solution and finally enter a neighborhood of the optimal controller $K^*$ and cost matrix $P^*$, respectively. The $\mathcal{H}_\infty$ norm of the closed-loop system is smaller than the threshold throughout the policy optimization. These numerical results are consistent with the developed theoretical results in Theorems \ref{thm:outerISS} and \ref{thm:innerISS}. 

Algorithm \ref{alg:IteAlg_learning} is implemented independently for $50$ trials to validate its performance. The length of the trajectory samples is $\tau = 5000$. The standard deviation of the system additive noise is $\Sigma = I_3$. The standard deviation of the exploratory noise is $\sigma_1 = \sigma_2 = 1$. The relative errors of the gain matrix and cost matrix are shown in Fig. \ref{fig:Norm_learn}. It is seen that the proposed off-policy RL algorithm can still approximate the optimal solution when the system is disturbed by additive Gaussian noise.


\begin{figure}[tb]
    \begin{subfigure}{.32\linewidth}
    \centering
      	\includegraphics [width=0.98\linewidth]{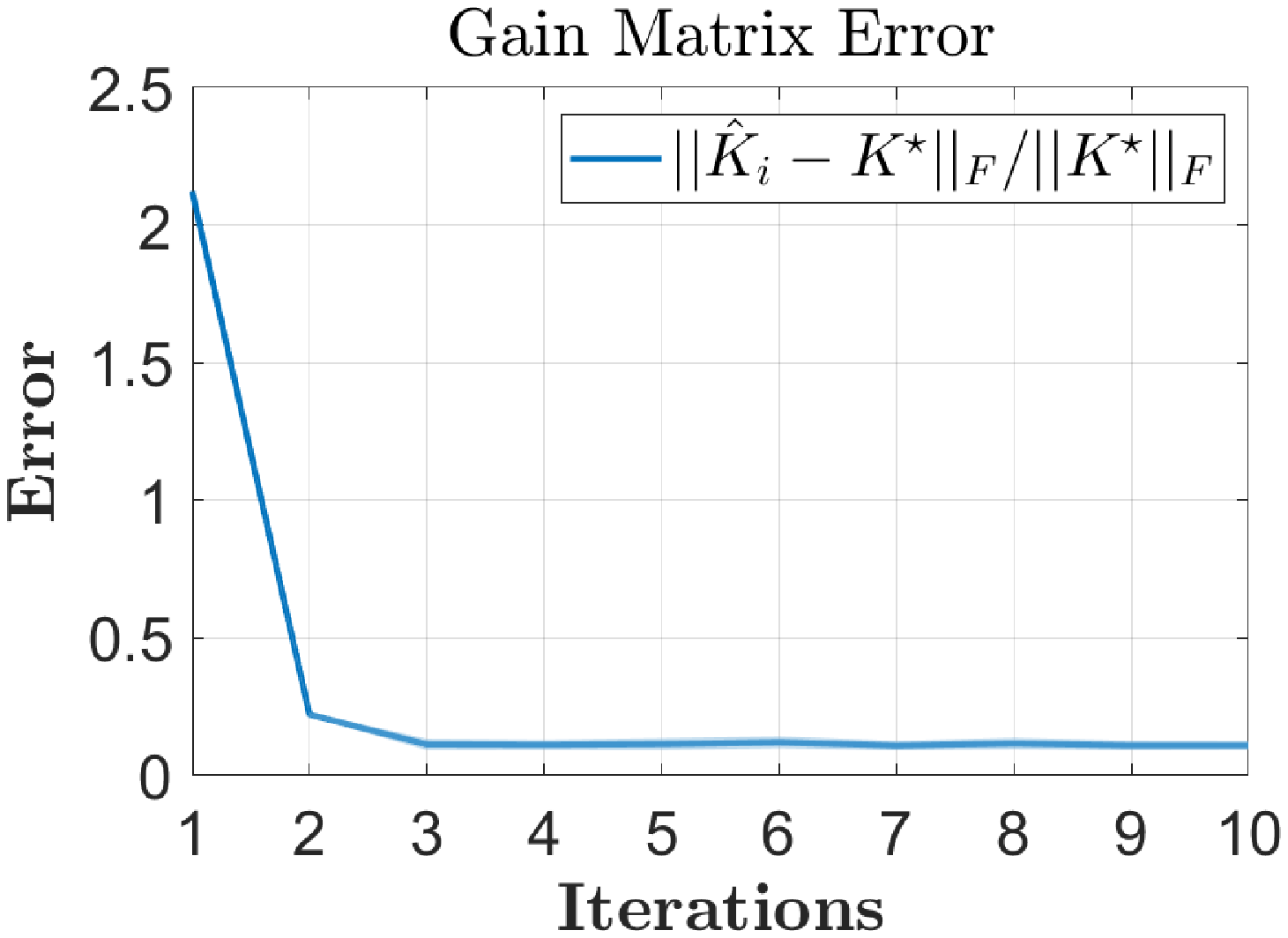}
    	\label{fig: K evolution}  
    \end{subfigure}
    \begin{subfigure}{.32\linewidth}
    	\centering
    	\includegraphics [width=0.98\linewidth]{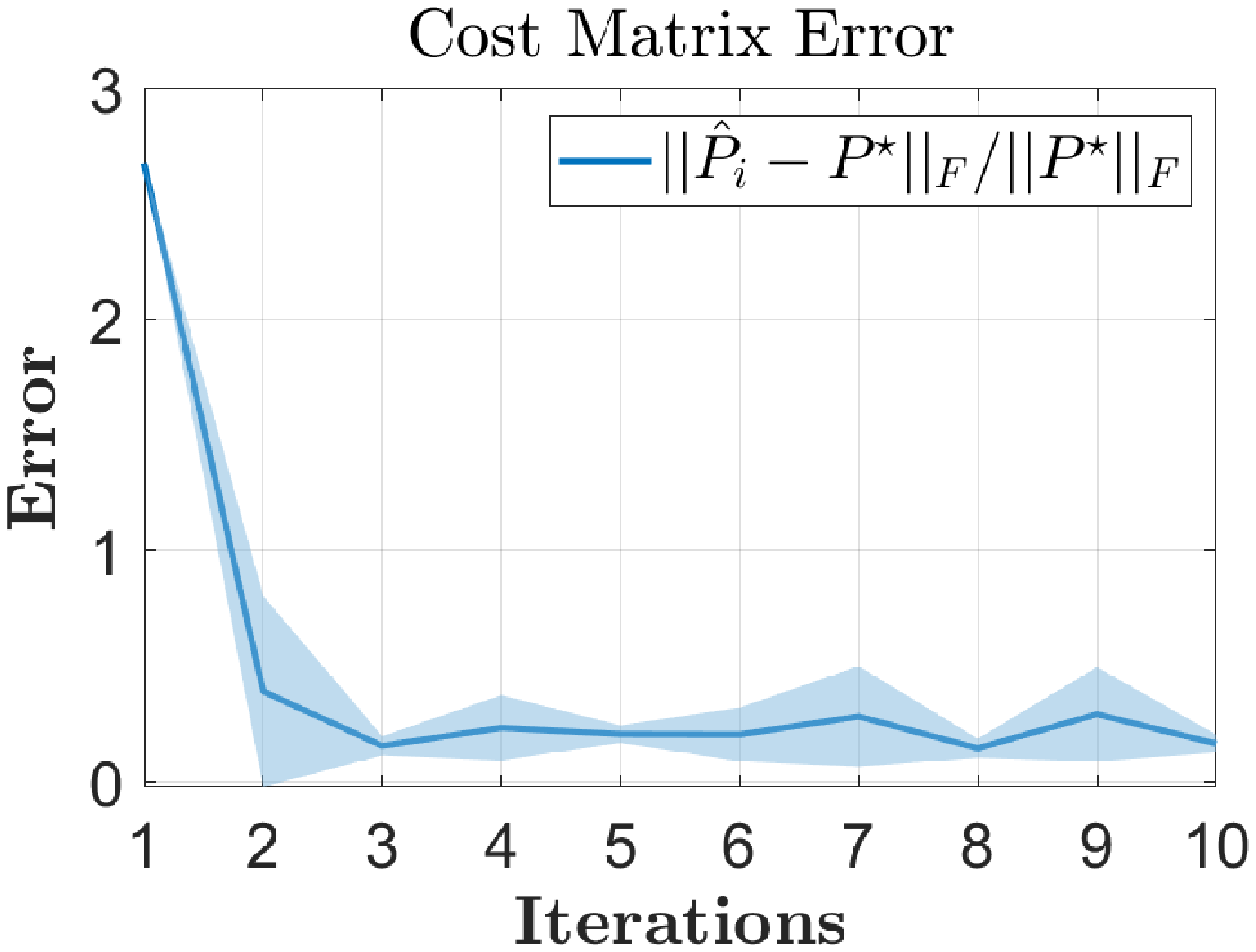}
    	\label{fig: P evolution}
    \end{subfigure}	
    \begin{subfigure}{.32\linewidth}
    	\centering
    	\includegraphics [width=0.98\linewidth]{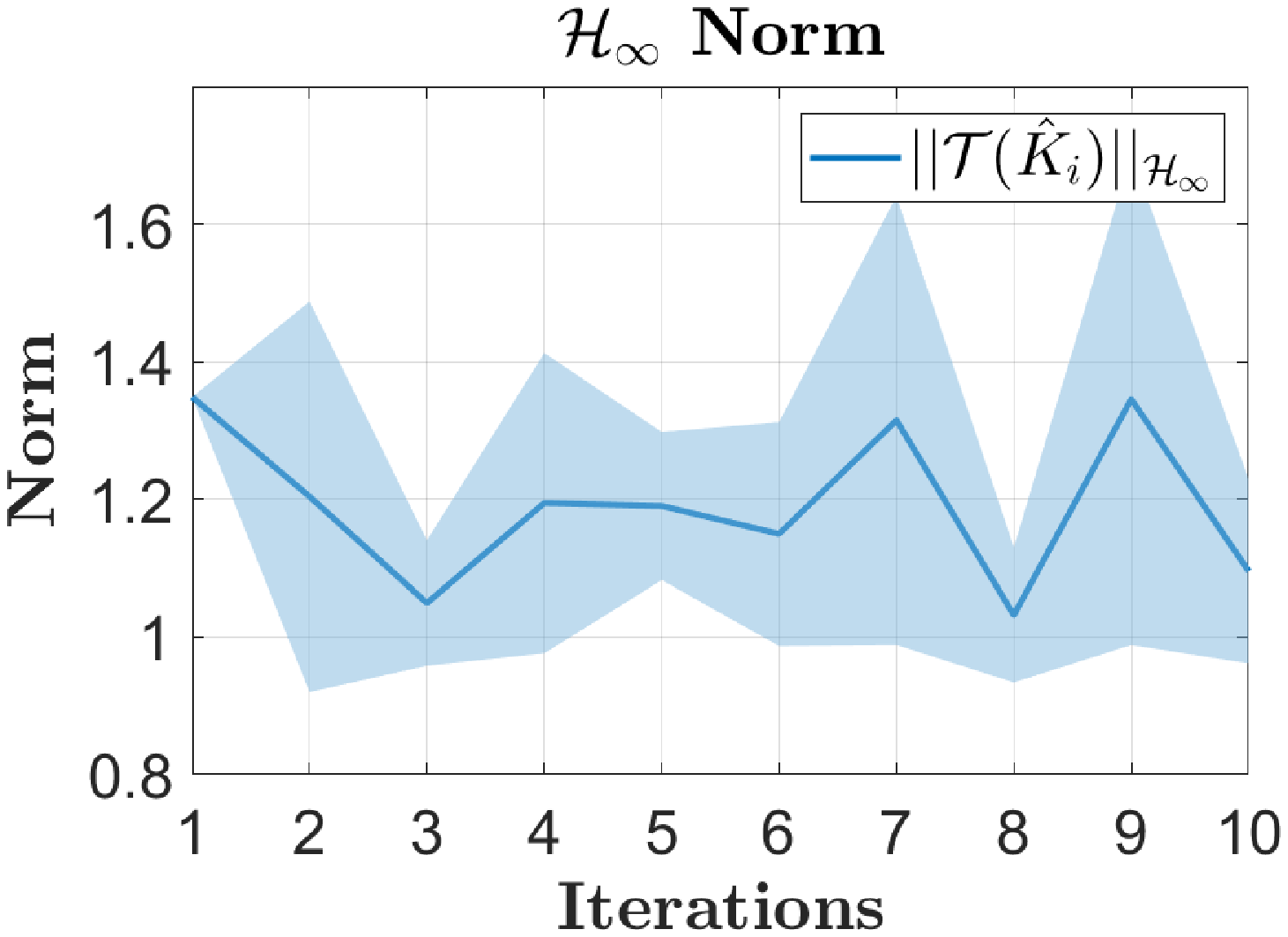}
    	\label{fig: H evolution}
    \end{subfigure}	
	\caption{Robustness of Algorithm \ref{alg:IteAlg_model} when $\norm{\Delta K}_\infty = 0.09$ and $\norm{\Delta L_i}_\infty = 0.09$.}
	\label{fig:Norm_PI}
\end{figure}

\begin{figure}[t]
	\centering
    \begin{subfigure}{.32\linewidth}
      	\includegraphics [width=0.98\linewidth]{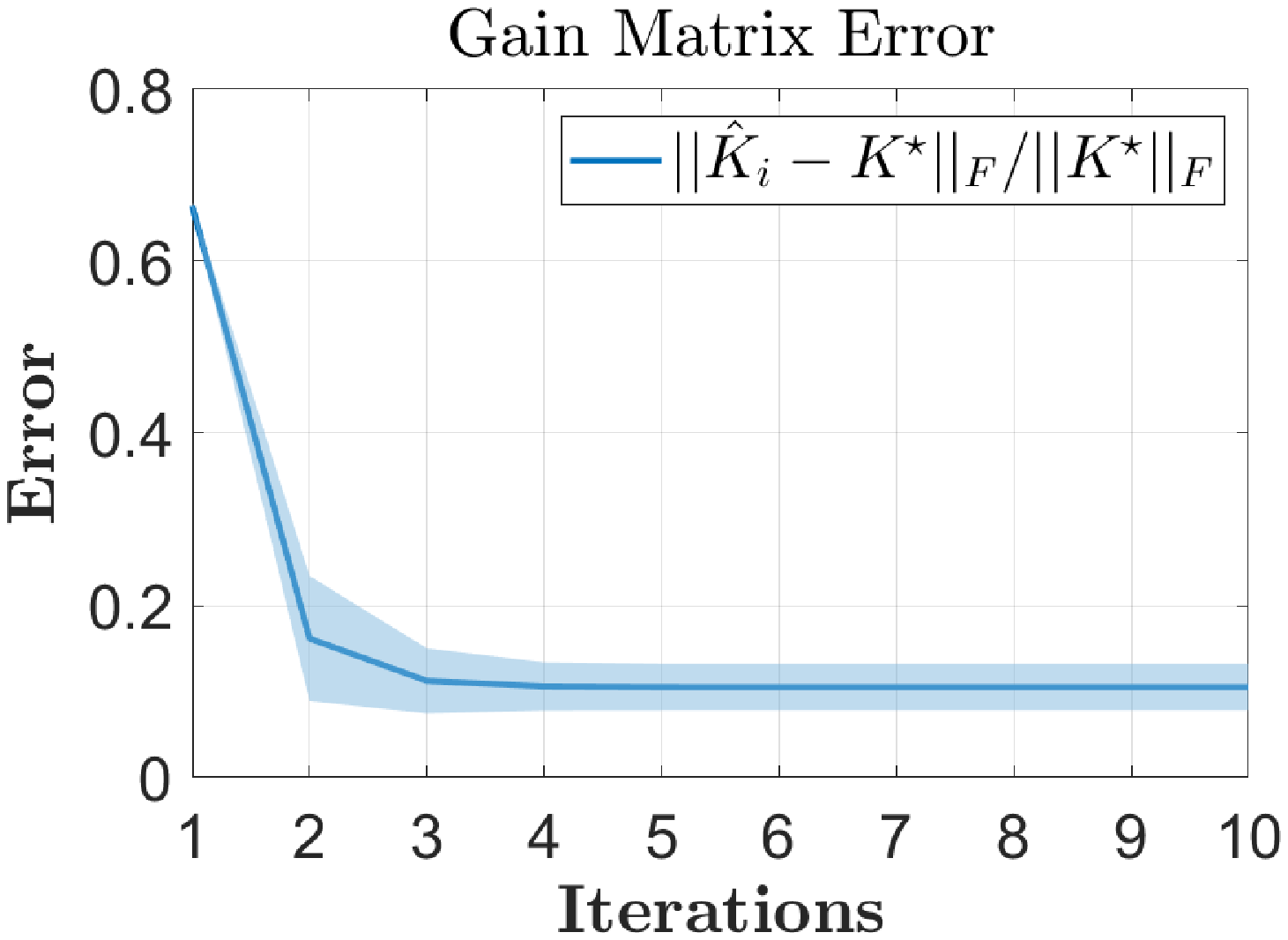}
    \end{subfigure}
    \begin{subfigure}{.32\linewidth}
    	\centering
    	\includegraphics [width=0.98\linewidth]{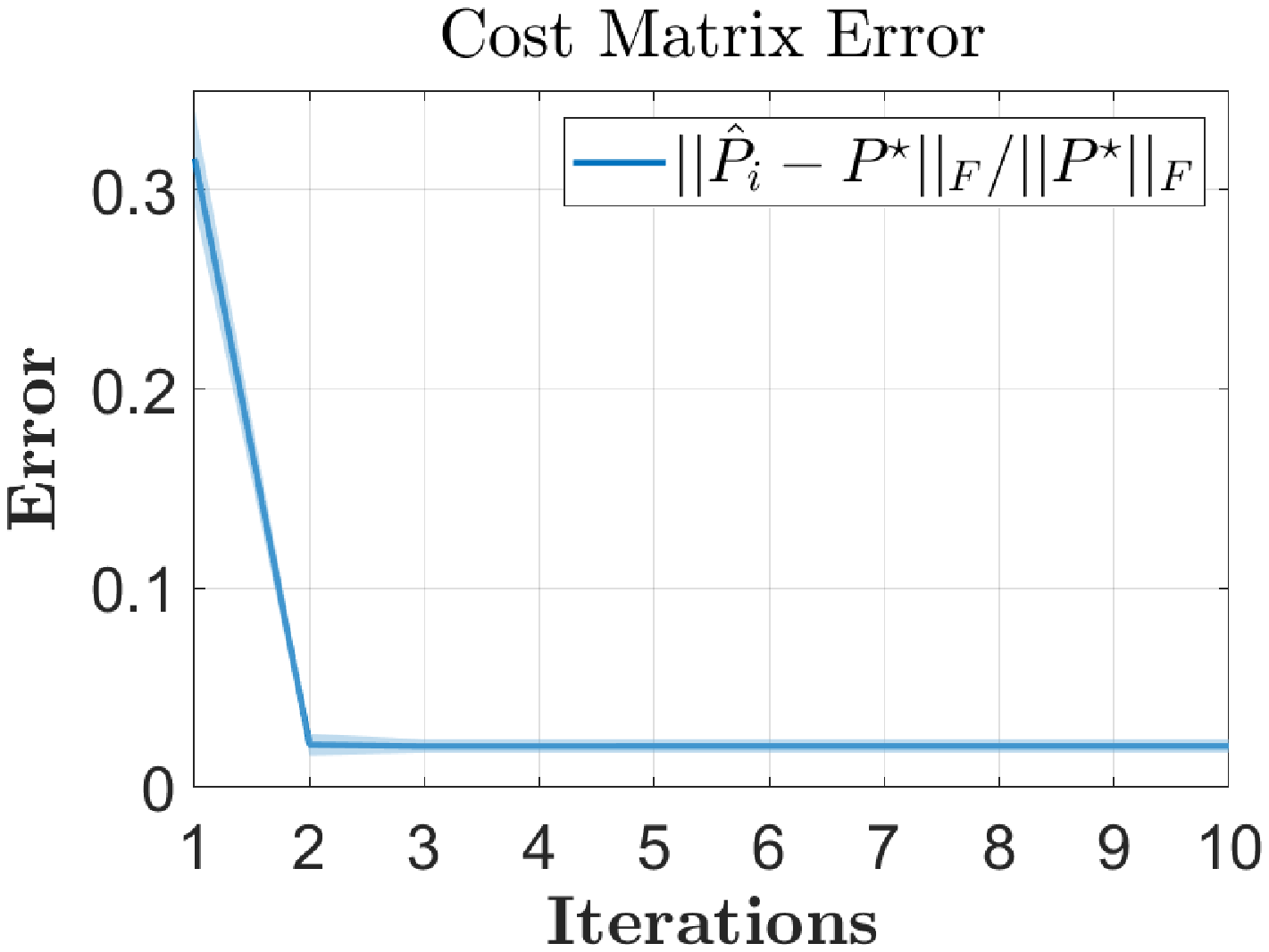}
    \end{subfigure}	
    \begin{subfigure}{.32\linewidth}
    	\centering
    	\includegraphics [width=0.98\linewidth]{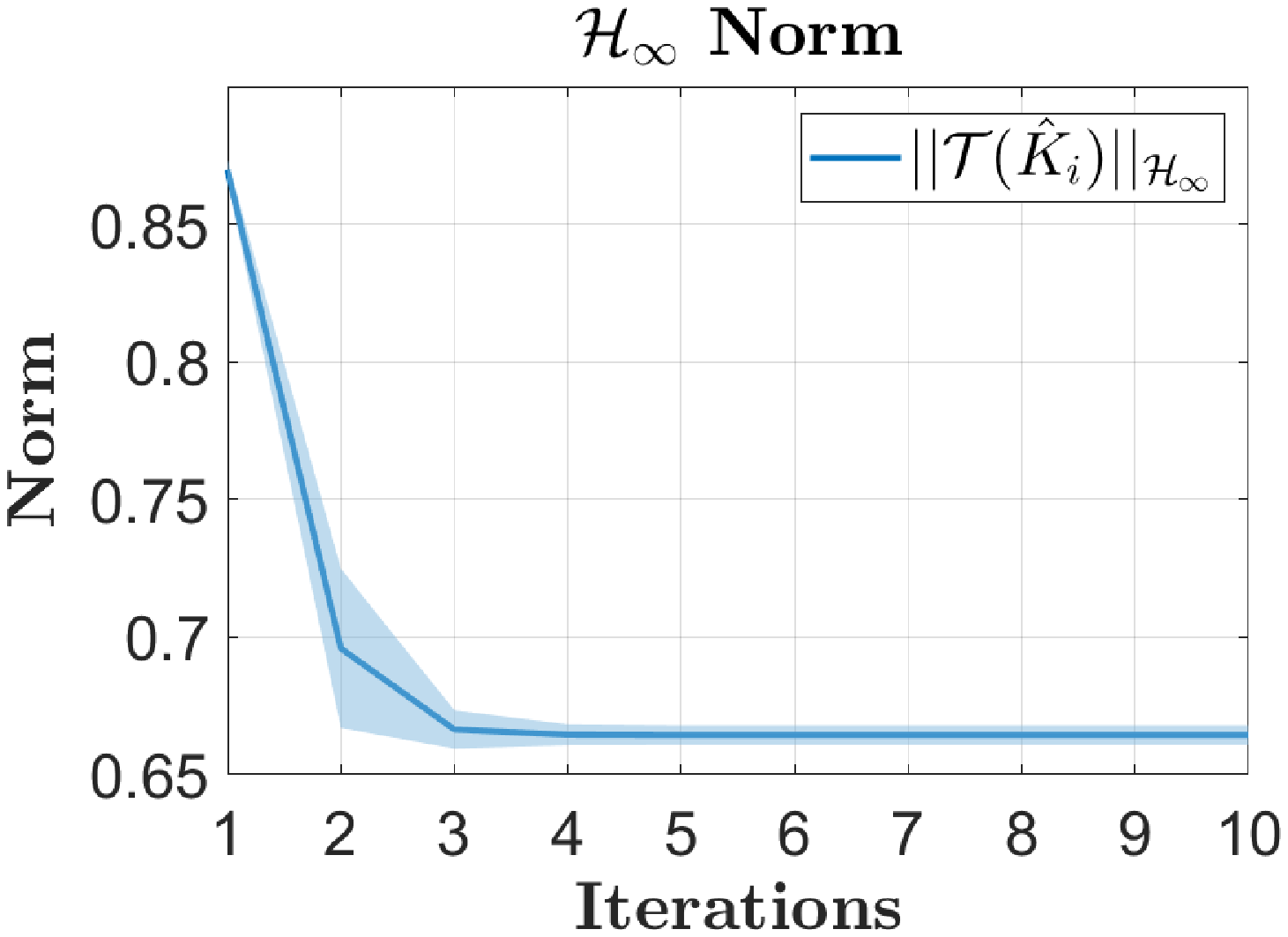}
    \end{subfigure}	
	\caption{Using Algorithm \ref{alg:IteAlg_learning}, the solutions of each iteration approach the optimal solution, and the $\mathcal{H}_\infty$ norm is smaller than the threshold.}
	\label{fig:Norm_learn}
\end{figure}

\subsection{Cart-Pole Example}
We next consider the cart-pole system in \cite{Anderson1989}, where the inverted pendulum is hinged to the top of a wheeled cart that moves along a straight line. The rotational angle is denoted as $\phi$ and the horizontal position is denoted by $s$. The mass of the cart is $m_c = 1kg$; the mass of the pendulum is $m_p = 0.1kg$; the distance from the pendulum's center of mass to the pivot is $l=0.5m$; the gravitational acceleration is $g=9.8m/sec^2$. The system is discretized under the sampling period $\Delta t=0.01 sec$. Considering the noise for the system, the system can be described by a linear system with $x = [s,\dot{s},\phi, \dot{\phi}]^T$, 
\begin{align}
    A = \begin{bmatrix} 1 & 0.01 &0 &0\\ 0 &1 &-0.01 &0\\ 0 & 0 & 1 &0.01 \\
    0 &0 &0.16 &1 \end{bmatrix}, B = \begin{bmatrix} 0 \\ 0.01 \\0 \\-0.015 \end{bmatrix}, \nonumber
\end{align}
$C=[I_4,0_{1 \times 4}]^T$, $D=0.001 I_4$, $E = [0_{1 \times 4}, 1]$. The $\mathcal{H}_\infty$ norm threshold is $\gamma = 10$.

The robustness evaluation of Algorithm \ref{alg:IteAlg_model} is shown in Fig. \ref{fig:CartPole_Norm_PI}, where the mean-variance curves are plotted for $50$ independent trials. At each iteration, the entries of the disturbances $\Delta K_i$ and $\Delta L_{i,j}$ are randomly sampled from a standard Gaussian distribution, and then $\norm{\Delta K_i}_F$ and $\norm{\Delta L_{i,j}}_F$ are normalized to $0.7$ and $0.1$, respectively. The relative errors approach zero even in the presence of the disturbance, demonstrating the small-disturbance ISS properties of the outer and inner loops in Theorems \ref{thm:outerISS} and \ref{thm:innerISS}. In addition, the $\mathcal{H}_\infty$ norm of the system is below the given threshold, and the robustness of the closed-loop system is guaranteed during the iteration.

When the matrices $(A,B,D)$ are unknown, Algorithm \ref{alg:IteAlg_learning} is implemented independently for $50$ times. The length of the sampled trajectory is $\tau = 10000$, the standard deviation of the system additive noise is $\Sigma = 0.1I_4$. The standard deviation of the exploratory noise is $\sigma_1 = \sigma_2 = 20$. It is seen from Fig. \ref{fig:CartPole_Norm_learn} that using the noisy data, the learning-based PO developed in Algorithm \ref{alg:IteAlg_learning} still finds a near-optimal solution.

\begin{figure}[tb]
	\centering
    \begin{subfigure}{.32\linewidth}
      	\includegraphics [width=0.99\linewidth]{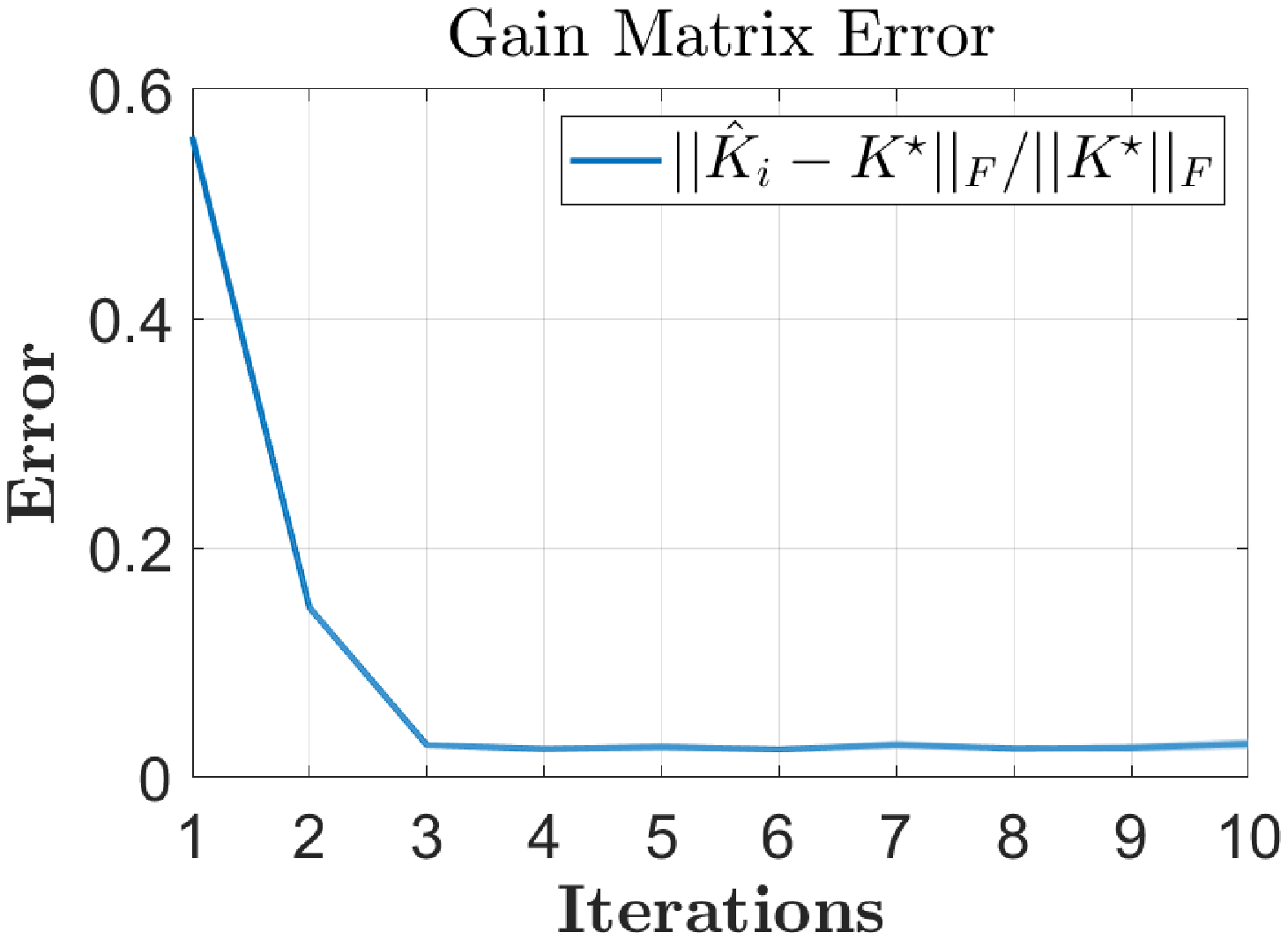}
    \end{subfigure}
    \begin{subfigure}{.32\linewidth}
    	\centering
    	\includegraphics [width=0.99\linewidth]{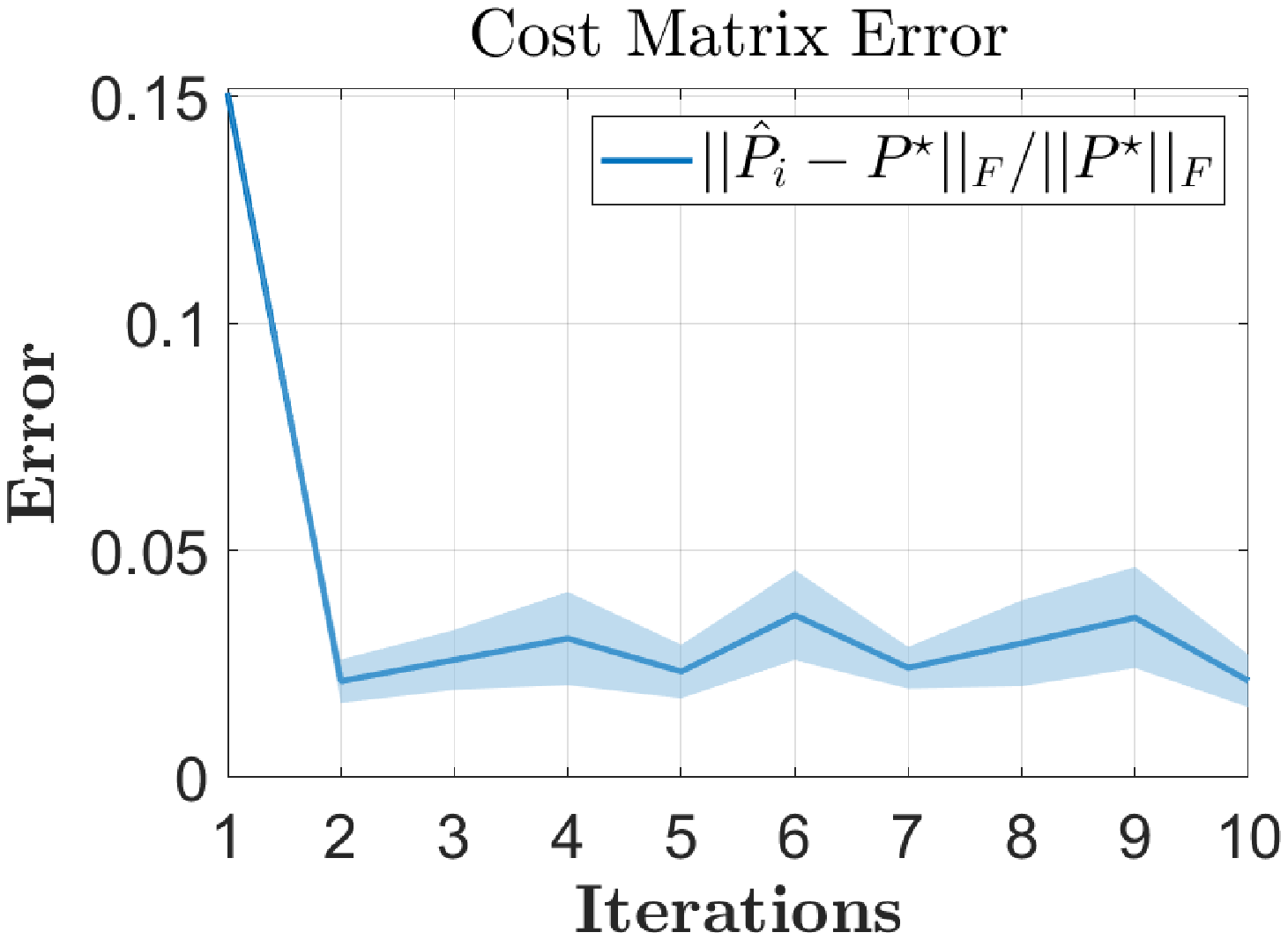}
    \end{subfigure}	
    \begin{subfigure}{.32\linewidth}
    	\centering
    	\includegraphics [width=0.99\linewidth]{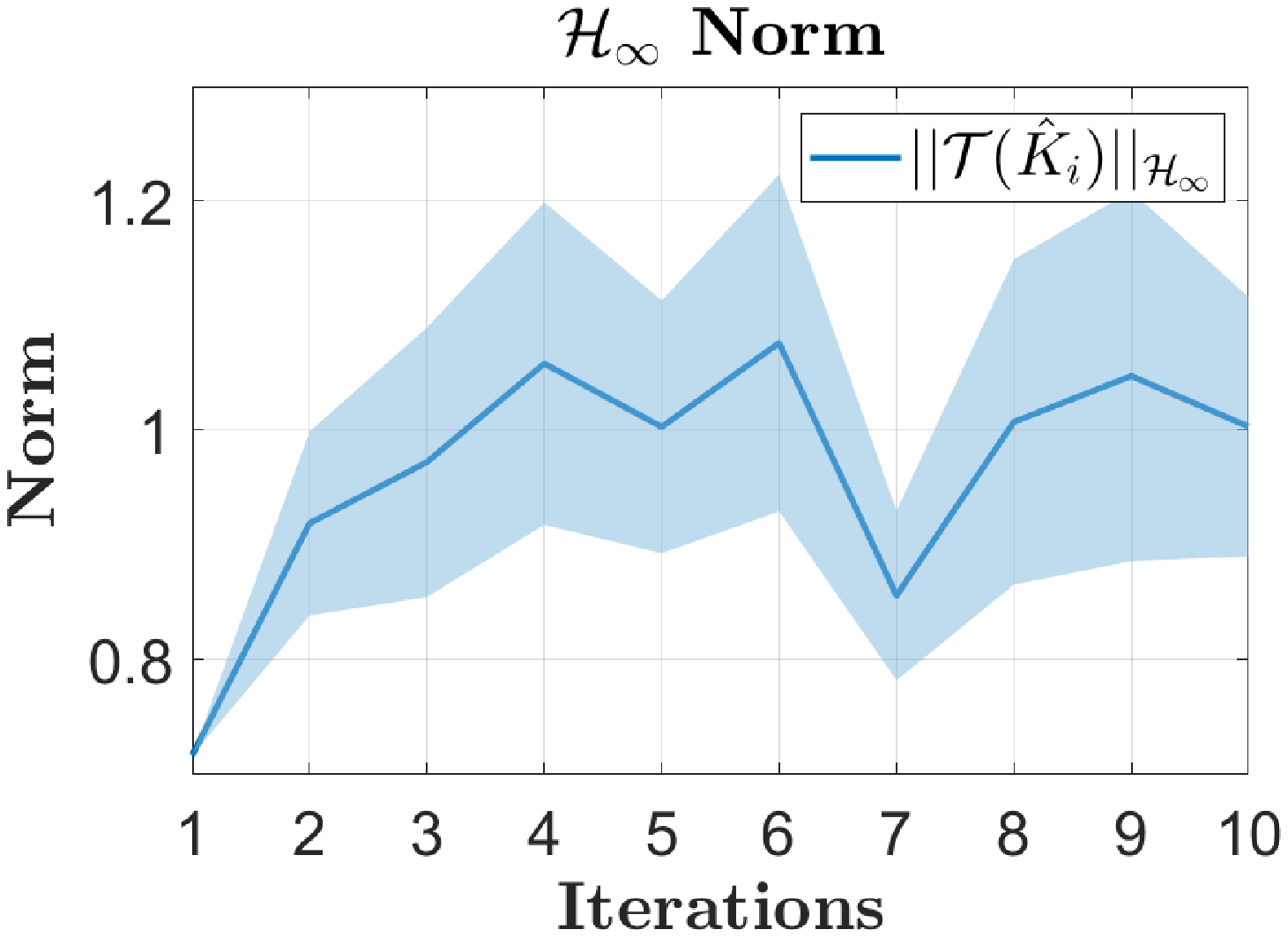}
    \end{subfigure}	
	\caption{For the cart-pole system, the robustness of Algorithm \ref{alg:IteAlg_model} when $\norm{\Delta K}_\infty = 0.7$ and $\norm{\Delta L_i}_\infty = 0.1$.}
	\label{fig:CartPole_Norm_PI}
\end{figure}

\begin{figure}[t]
	\centering
    \begin{subfigure}{.32\linewidth}
      	\includegraphics [width=0.98\linewidth]{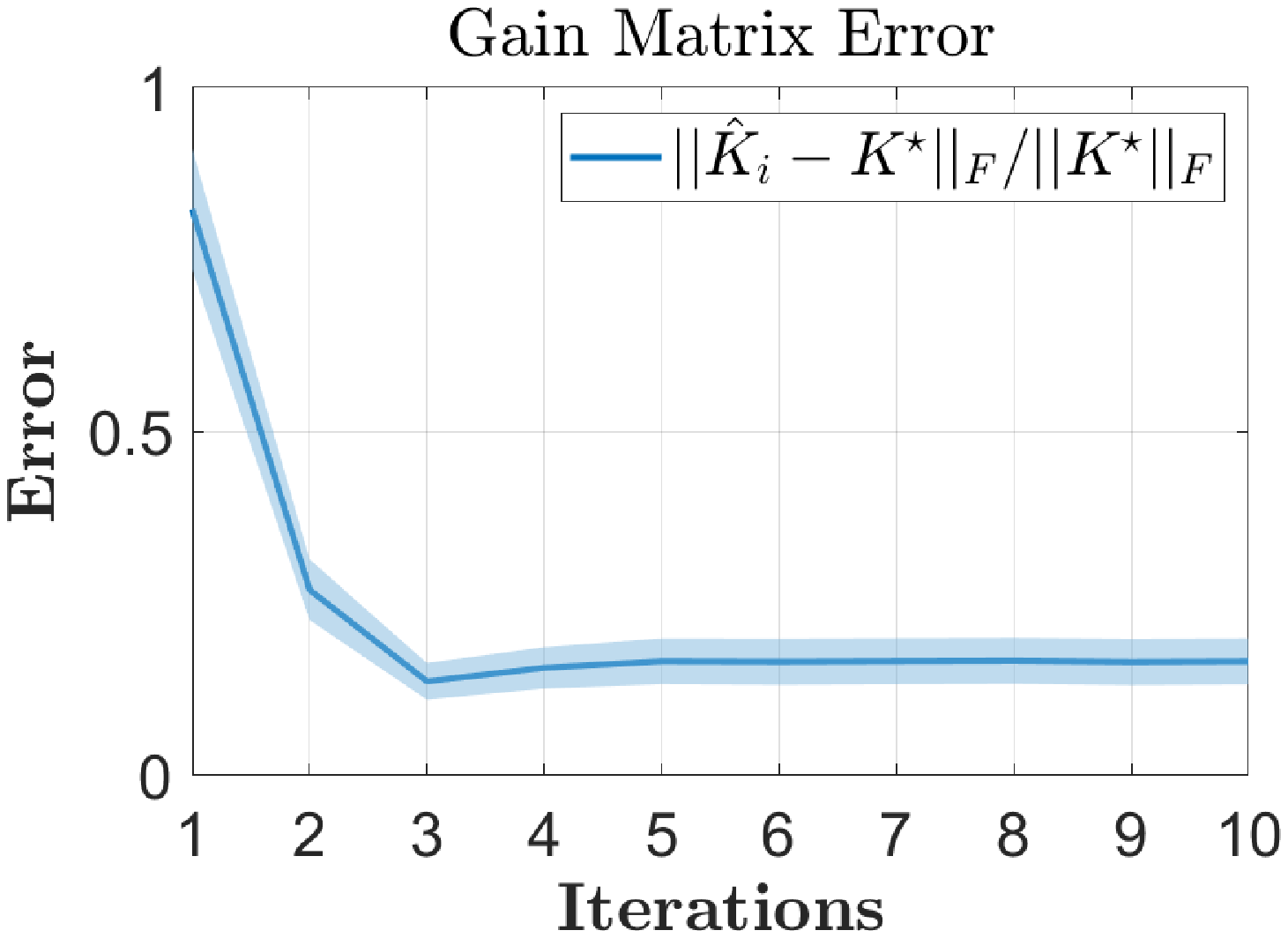}
    \end{subfigure}
    \begin{subfigure}{.32\linewidth}
    	\centering
    	\includegraphics [width=0.98\linewidth]{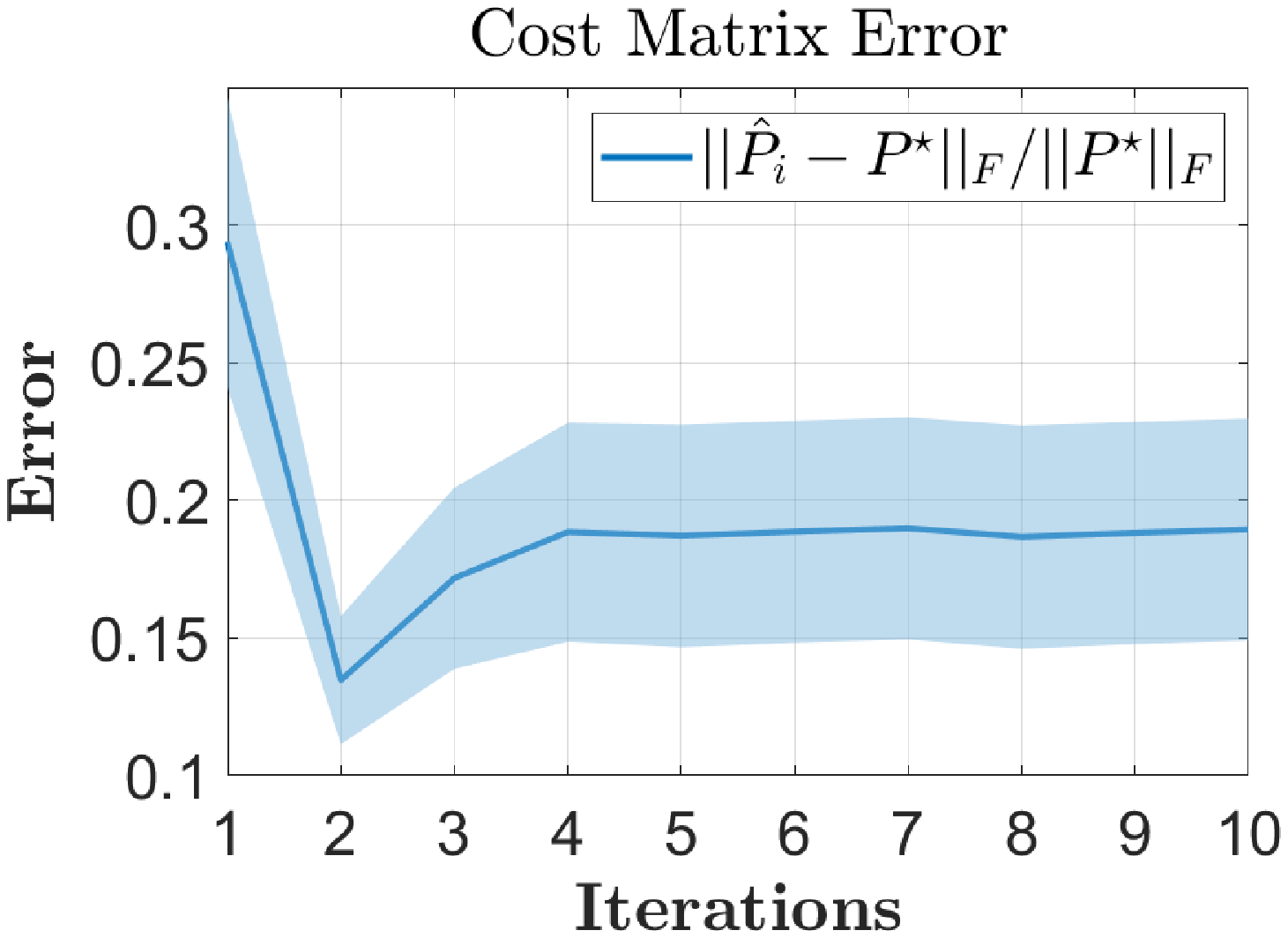}
    \end{subfigure}	
    \begin{subfigure}{.32\linewidth}
    	\centering
    	\includegraphics [width=0.98\linewidth]{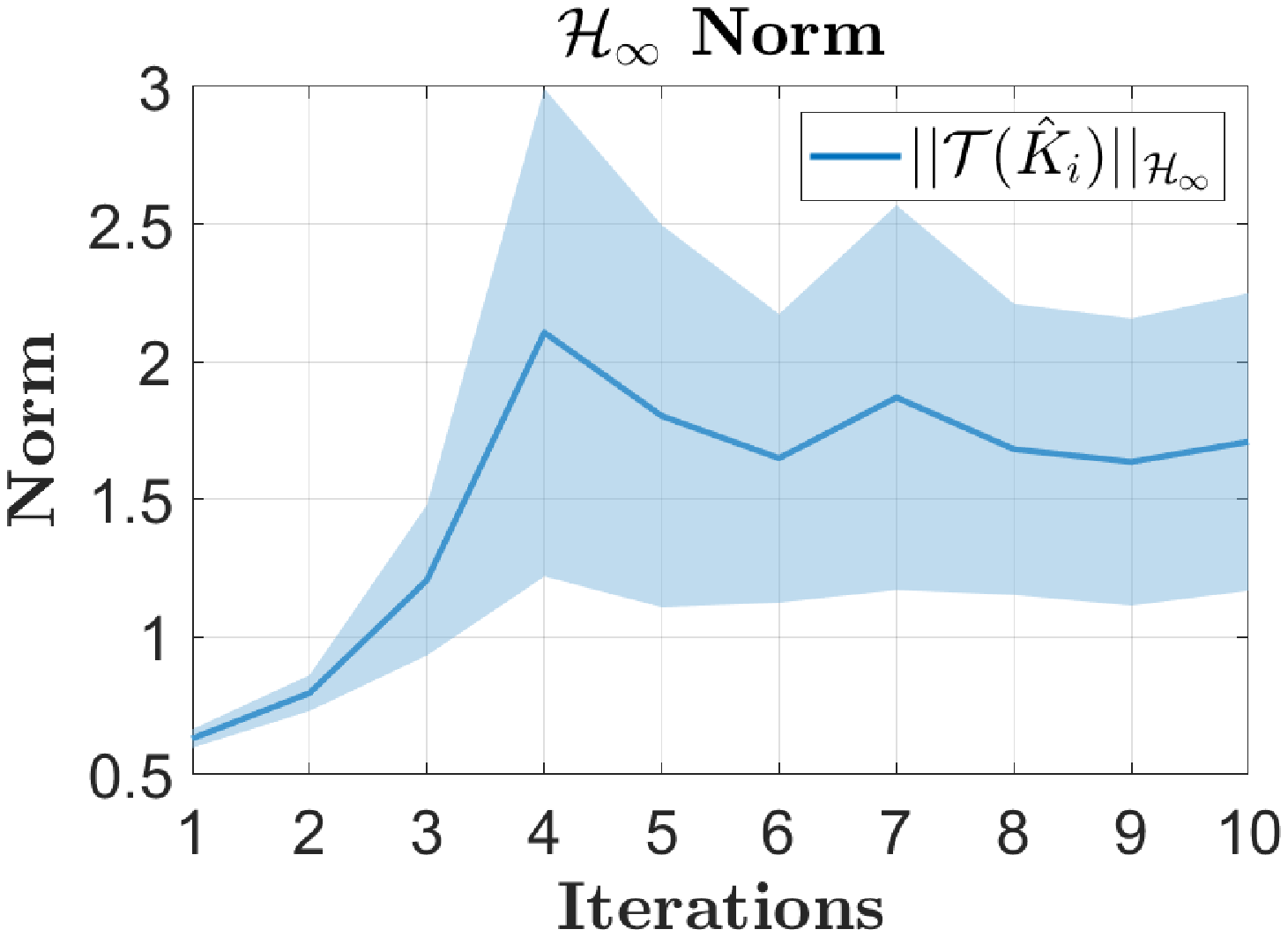}
    \end{subfigure}	
	\caption{For the cart-pole system, as the iteration of Algorithm \ref{alg:IteAlg_learning} proceeds, the gain and cost matrices approach the optimal solution, and the $\mathcal{H}_\infty$ norm is smaller than the threshold.}
	\label{fig:CartPole_Norm_learn}
\end{figure}

\section{Conclusion}

In this paper, we have proposed a novel dual-loop policy optimization algorithm for data-driven risk-sensitive linear quadratic Gaussian control whose convergence and robustness properties have been analyzed. We have shown that the iterative algorithm possesses the property of small-disturbance input-to-state stability, that is, starting from any initial admissible controller, the solutions of the proposed policy optimization algorithm ultimately enter a neighbourhood of the optimal solution, given that the disturbance is relatively small. Based on these model-based theoretical results, when the accurate system knowledge is unavailable, we have also proposed a novel off-policy policy optimization RL algorithm to learn from data robust optimal controllers. Numerical examples are provided and the efficacy of the proposed methods are demonstrated by an academic example and a benchmark cart-pole system. Future work will be directed toward investigating the input-to-state stability of standard gradient descent and natural gradient descent algorithms. In addition, by invoking small-gain theory, the application of learning-based risk-sensitive LQG control to decentralized control design of complex systems will be studied.

\newcounter{appidx}
\setcounter{appidx}{1}
\renewcommand\theequation{\Alph{appidx}.\arabic{equation}}

\setcounter{equation}{0}

\section*{Appendix \Alph{appidx}: Useful Auxiliary Results}
Some fundamental lemmas are introduced to assist in the development of the results in the main body of the paper.
\begin{lemma}\cite[Theorem 5.D6]{book_chen}\label{lm:lyapunov}
Consider a stable matrix $A \in \mathbb{R}^{n \times n}$ satisfying 
\begin{align}\label{eq:lya}
    A^TPA - P + Q = 0.
\end{align}
The following statements hold:
\begin{enumerate}
    \item $P = \sum_{t=0}^\infty (A^T)^t Q A^t$ is the unique solution to \eqref{eq:lya} and $P = P^T \succeq 0$ if $Q \succeq 0$;
    \item For any other $P' = P'^T \succeq 0$ satisfying $ A^TP'A - P' + Q' = 0$ with $Q' \succeq Q$, $P' \succeq P$.
\end{enumerate}
\end{lemma}


\begin{lemma}[Bounded Real Lemma]\label{lm:realbounded}
Given a matrix $K \in 
\mathbb{R}^{m \times n}$ and the transfer function $\mathcal{T}(K)$ for the linear system \eqref{eq:LTI} under Assumptions 1 and 2, the following statements are equivalent:\\
    1) $K$ is stabilizing and $\norm{\mathcal{T}(K)}_{\mathcal{H}_\infty} < \gamma$;\\
    2) The algebraic Riccati equation \eqref{eq:AREforKPrelimi}
    admits a unique stabilizing solution $P_K=P_K^T \succ 0$ such that i) $I_q - \gamma^{-2}D^TP_KD \succ 0$; ii) $[A-BK + D(\gamma^2I_q - D^TP_KD)^{-1}D^TP_K(A-BK)]$ is stable;\\
    3) There exists $P_K=P_K^T \succ 0$ such that 
    \begin{align}\label{eq:AREineq}
        &I_q - \gamma^{-2}D^TP_KD \succ 0, \\
        &(A-BK)^TU_K(A-BK) - P_K + Q + K^TRK \prec 0. \nonumber
    \end{align} \\
    4) There exist $W = W^T \succ 0$ and $V = KW$, such that
    \begin{align} \label{eq:lmiforhinfpaper}
    &\begin{bmatrix}
    -W &* &* &* \\
    0  &-\gamma^2 I_q &* &* \\
    AW - BV &D &-W &* \\
    CW - EV &0 &0 &-I_p
    \end{bmatrix} \prec 0.
    \end{align}
\end{lemma}
\begin{proof}
    The equivalence of the first three statements is from \cite[Lemma 2.7]{zhangarxiv2019}. We next prove the statement 4). By the Schur complement lemma, \eqref{eq:lmiforhinfpaper} is equivalent to 
    \begin{subequations}\label{eq:schurComp}
    \begin{align} 
        &-W - \begin{bmatrix} 0 &WA^T - V^TB^T &WC^T - V^TE^T \end{bmatrix}
        \begin{bmatrix} -\gamma^2 I_q &D^T &0 \\
        D &-W &0 \\ 
        0 &0 &-I_p\end{bmatrix}^{-1}
        \begin{bmatrix} 0 \\ AW-BV \\ CW-EV\end{bmatrix} \prec 0, \label{eq:schurComp1} \\
        &\begin{bmatrix} -\gamma^2 I_q &D^T &0 \\
        D &-W &0 \\ 
        0 &0 &-I_p\end{bmatrix} \prec 0. \label{eq:schurComp2}
    \end{align}    
    \end{subequations}
    Pre-multiplying and post-multiplying \eqref{eq:schurComp1} by $P=W^{-1}$, and applying Schur complement on \eqref{eq:schurComp2}, it is found that \eqref{eq:schurComp} is equivalent to
    \begin{subequations}\label{eq:schurCompwithP}
    \begin{align} 
        &-P - \begin{bmatrix} 0 &(A-BK)^T &(C - EK)^T \end{bmatrix}
        \begin{bmatrix} -\gamma^2 I_q &D^T &0 \\
        D &-W &0 \\ 
        0 &0 &-I_p\end{bmatrix}^{-1}
        \begin{bmatrix} 0 \\ A-BK \\ C - EK\end{bmatrix} \prec 0, \\
        &-\gamma^2 I_q + D^TPD \prec 0. 
    \end{align}    
    \end{subequations}
    By recalling $(C-EK)^T(C-EK) = Q + K^TRK$, it follows that \eqref{eq:schurCompwithP} is equivalent to \eqref{eq:AREineq}. Thus, the proof is completed.      
\end{proof}
\begin{lemma}\label{lm:normInequ}
For any positive semi-definite matrix $P \in \mathbb{S}^n$, $\norm{P}_F \leq \Tr(P) \leq \sqrt{n} \norm{P}_F$ and $\norm{P} \leq \Tr(P) \leq n \norm{P}$.
\end{lemma}
\begin{proof}
    Let $\sigma_1 \geq \cdots \geq \sigma_n$ denote the descending sequence of $P$'s singular values. Then, $\norm{P}_F = \sqrt{\sum_{i=1}^n \sigma_i^2}$, $\Tr(P) = \sum_{i=1}^n \sigma_i$, and $\norm{P} = \sigma_1(P)$. Hence, $\norm{P} \leq \Tr(P) \leq n \norm{P}$. $\norm{P}_F \leq \Tr(P) \leq \sqrt{n} \norm{P}_F$ can be proven by the fact that $\sum_{i=1}^n \sigma_i^2 \leq (\sum_{i=1}^n \sigma_i)^2 \leq n\sum_{i=1}^n \sigma_i^2$.
\end{proof}

\begin{lemma}\label{lm:ComDupMat}
 There exists a unique matrix $T_{n} \in \mathbb{R}^{n^2 \times \frac{1}{2}n(n+1)}$ with full column rank, such that for every $X \in \mathbb{S}^{n}$,
\begin{align}\label{eq:DupMat}
    \vect(X) = T_{n} \vecs(X), \quad \vecs(X) = T_{n}^{\dagger}\vect(X).
\end{align}
\end{lemma}
\begin{proof}
  See \cite[pp. 56]{book_Magnus} where $T_n$ is called duplication matrix. 
\end{proof}

\stepcounter{appidx}
\setcounter{equation}{0}

\section*{Appendix \Alph{appidx}: Proof of Theorem \ref{thm:outerloopGloballinear}}

The following lemma shows that the cost matrix $P_i$ generated by the outer-loop iteration \eqref{eq:outerloopIte} is monotonically decreasing and all the updated feedback gains are admissible given an initial admissible feedback gain. 
\begin{lemma}\label{lm:outerloop_converge}
Under Assumptions \ref{ass:stabilizable} and \ref{ass:crossterm}, if $K_1 \in \mathcal{W}$, then for any $i \in \mathbb{Z}_+$, 
\begin{enumerate}
    \item $K_i \in \mathcal{W}$;
    \item $P_{1} \succeq \cdots \succeq P_{i} \succeq P_{i+1} \succeq \cdots \succeq P^*$;
    \item $\lim_{i \to \infty}\norm{K_i - K^*}_F = 0$ and $\lim_{i \to \infty}\norm{P_i - P^*}_F = 0$.
\end{enumerate}
\end{lemma}
\begin{proof}
The statements 1) and 3) are shown in \cite[Theorem 4.3 and Theorem 4.6]{zhangarxiv2019}. The statement 2) follows from \cite[Equation (5.27)]{zhangarxiv2019}.
\end{proof}

The following lemma presents the difference between the updated controller $K'$ and the optimal controller $K^*$. 
\begin{lemma}\label{lm:K-Kopt}
    For any $K \in \mathcal{W}$ and $K' := (R+B^TU_KB)^{-1}B^TU_KA$, we have
    \begin{align}
        &K' - K^* = R^{-1}B^T (\Lambda_K)^{-T}(P_K - P^*)A^*,\\
        &\Lambda_K := I_n + BR^{-1}B^TP_K  - \gamma^{-2} DD^TP_K
    \end{align}
\end{lemma}
\begin{proof}
By the expression of $U_K$ in \eqref{eq:UKExpre}, we have
\begin{align}
\begin{split}
    &(R+B^TU_KB)^{-1}B^T U_K A = R^{-1}(I_m+B^TU_KBR^{-1})^{-1}B^TU_KA \\
    &\qquad = R^{-1}B^TU_K(I_n+BR^{-1}B^TU_K)^{-1}A \\
    & \qquad = R^{-1}B^TU_K[I_n + BR^{-1}B^TP_K +  BR^{-1}B^TP_KD(\gamma^2I_q - D^TP_KD)^{-1}D^TP_K ]^{-1}A \\
    & \qquad = R^{-1}B^TU_K\{I_n + BR^{-1}B^TP_K[I_n + DD^TP_K(\gamma^{2}I_n - DD^TP_K)^{-1} ] \}^{-1}A \\
    & \qquad = R^{-1}B^TU_K[I_n + BR^{-1}B^TP_K( I_n - \gamma^{-2} DD^TP_K)^{-1}]^{-1}A  \\
    & \qquad = R^{-1}B^TU_K ( I_n - \gamma^{-2} DD^TP_K) (I_n + BR^{-1}B^TP_K  - \gamma^{-2} DD^TP_K)^{-1}A  \\
    & \qquad = R^{-1}B^T P_K [I_n + D(\gamma^{2}I_q -D^TP_KD)^{-1}D^TP_K] ( I_n - \gamma^{-2} DD^TP_K) \Lambda_K^{-1}A  \\
    & \qquad = R^{-1}B^T P_K [I_n + \gamma^{-2}DD^TP_K(I_n -\gamma^{-2}DD^TP_K)^{-1}] ( I_n - \gamma^{-2} DD^TP_K) \Lambda_K^{-1}A \\
    & \qquad = R^{-1}B^T P_K \Lambda_K^{-1}A. 
\end{split}
\end{align}
By the expressions of $K'$ and $K^*$ in $\eqref{eq:Kopt}$, we have
\begin{align} \label{eq:K-Kopt}
    &K' - K^* =  R^{-1}B^TP_K\Lambda_K^{-1} A - R^{-1}B^TP^*(\Lambda^*)^{-1}A \\
    & \quad = R^{-1}B^T (P_K - P^*) (\Lambda^*)^{-1} A + R^{-1}B^TP_K \Lambda_{K}^{-1}A - R^{-1}B^TP_K (\Lambda^*)^{-1} A \nonumber\\
    & \quad = R^{-1}B^T (P_K - P^*) (\Lambda^*)^{-1} A + R^{-1}B^TP_K \left[ \Lambda_{K}^{-1} \Lambda^* (\Lambda^*)^{-1} -  \Lambda_{K}^{-1} \Lambda_{K}  (\Lambda^*)^{-1}\right] A \nonumber\\
    & \quad=  R^{-1}B^T (P_K - P^*) (\Lambda^*)^{-1} A - R^{-1}B^TP_K\Lambda_K^{-1} ( BR^{-1}B^T - \gamma^{-2}DD^T)(P_K - P^*)  (\Lambda^*)^{-1} A \nonumber\\
    &\quad = R^{-1}B^T[I_n - P_K\Lambda_K^{-1} ( BR^{-1}B^T - \gamma^{-2}DD^T)](P_K - P^*) A^* \nonumber\\
    &\quad = R^{-1}B^T[I_n - (I_n + P_KBR^{-1}B^T - \gamma^{-2}P_KDD^T)^{-1} P_K( BR^{-1}B^T - \gamma^{-2}DD^T)](P_K - P^*)  A^* \nonumber\\
    &\quad = R^{-1}B^T(I_n + P_KBR^{-1}B^T - \gamma^{-2}P_KDD^T)^{-1} (P_K - P^*) A^* = R^{-1}B^T (\Lambda_K)^{-T}(P_K - P^*)A^*. \nonumber
\end{align}
It is noticed that the fifth equation follows from the equality
\begin{align}\label{eq:AoptReexp}
    A^* &= [I_n + D(\gamma^2I_q - D^TP^*D)^{-1}D^TP^*](A - BK^*) = (I_n - \gamma^{-2}DD^TP^*)^{-1}(A - BK^*) \nonumber\\
    &= (I_n - \gamma^{-2}DD^TP^*)^{-1}[I_n - B(R+B^TU^*B)^{-1}B^TU^*]A \nonumber\\
    &= (I_n - \gamma^{-2}DD^TP^*)^{-1}(I_n + BR^{-1}B^TU^*)^{-1}A \\
    &= (I_n - \gamma^{-2}DD^TP^*)^{-1}[I_n + BR^{-1}B^TP^*(I_n - \gamma^{-2}DD^TP^*)^{-1}]^{-1}A = (\Lambda^*)^{-1} A. \nonumber
\end{align}
\end{proof}

The following lemma presents the expression of the difference between $U_K$ and $U^*$.
\begin{lemma}\label{lm:UiUoptQuad}
For any $K \in \mathcal{W}$, $(U_K - U^*)$ satisfies
\begin{align}\label{eq:Ui-U*Lemma}
    &U_K - U^* \nonumber\\
    &= (I_n - \gamma^{-2}P^*DD^T)^{-1}(P_K - P^*) (I_n - \gamma^{-2}DD^TP^*)^{-1} \nonumber\\
    &+ (I_n - \gamma^{-2}P^*DD^T)^{-1}(P_K - P^*)D(\gamma^2I_q - D^TP_KD)^{-1} \nonumber\\
    &\quad D^T(P_K - P^*)(I_n - \gamma^{-2}DD^TP^*)^{-1} 
\end{align}
\end{lemma}
\begin{proof}
By recalling the expression of $U^*$ in \eqref{eq:GARE}, we can further derive it as
\begin{align}\label{eq:UoptLemma}
    &U^*= [I_n + P^*D(\gamma^{2}I_q - D^TP^*D)^{-1}D^T]P^* = [I_n + P^*DD^T(\gamma^{2}I_n - P^*DD^T)^{-1}]P^* \nonumber\\ 
    &= (I_n - \gamma^{-2}P^*DD^T)^{-1}P^* = (I_n - \gamma^{-2}P^*DD^T)^{-1}P^*(I_n - \gamma^{-2}DD^TP^*)(I_n - \gamma^{-2}DD^TP^*)^{-1} \nonumber \\
    &=  (I_n - \gamma^{-2}P^*DD^T)^{-1}(P_K - \gamma^{-2}P^*DD^TP^*)(I_n - \gamma^{-2}DD^TP^*)^{-1} \nonumber\\
    &-  (I_n - \gamma^{-2}P^*DD^T)^{-1}(P_K - P^*)(I_n - \gamma^{-2}DD^TP^*)^{-1}.
\end{align}
To simplify the notation, define $S_K$ as
\begin{align}
    S_K := U_K - (I_n - \gamma^{-2}P^*DD^T)^{-1}(P_K - \gamma^{-2}P^*DD^TP^*)(I_n - \gamma^{-2}DD^TP^*)^{-1}.
\end{align}
Then, it follows that
\begin{align}\label{eq:Si}
    &(I_n - \gamma^{-2}P^*DD^T)S_K(I_n - \gamma^{-2}DD^TP^*) \\
    &= (I_n - \gamma^{-2}P^*DD^T)U_K(I_n - \gamma^{-2}DD^TP^*) - (P_K - \gamma^{-2}P^*DD^TP^*) \nonumber\\
    & = U_K - P_K - \gamma^{-2}P^*DD^TU_K - \gamma^{-2}U_KDD^TP^* + \gamma^{-4}P^*DD^TU_KDD^TP^* + \gamma^{-2}P^*DD^TP^*. \nonumber
\end{align}
Following the expression of $U_K$ in \eqref{eq:UKExpre}, we have
\begin{align}\label{eq:PDDU}
    &\gamma^{-2}P^*DD^TU_K \\
    &= \gamma^{-2}P^*DD^T(I_n - \gamma^{-2}P_KDD^T)^{-1}P_K =  \gamma^{-2}P^*D(I_q - \gamma^{-2}D^TP_KD)^{-1}D^TP_K. \nonumber
\end{align}
Since $U_K - P_K = P_KD(\gamma^2I_q - D^TP_KD)^{-1}D^TP_K$, plugging \eqref{eq:PDDU} into \eqref{eq:Si} and completing the squares yield
\begin{align}\label{eq:Si2}
    &(I_n - \gamma^{-2}P^*DD^T)S_K(I_n - \gamma^{-2}DD^TP^*) = (P_K - P^*)D(\gamma^2I_q - D^TP_KD)^{-1}D^T(P_K - P^*) \nonumber\\
    &\qquad - P^*D(\gamma^2I_q - D^TP_KD)^{-1}D^TP^*   + \gamma^{-4}P^*DD^TU_KDD^TP^* + \gamma^{-2}P^*DD^TP^*.
\end{align}
By the matrix inversion lemma,
\begin{align} \label{eq:matrixinverse}
    (\gamma^2I_q - D^TP_KD)^{-1} = \gamma^{-2}I_q + \gamma^{-4} D^T (I_n - \gamma^{-2}P_KDD^T)^{-1} P_KD = \gamma^{-2}I_q + \gamma^{-4} D^T U_KD,
\end{align}
where the last equality is from $U_K =  (I_n - \gamma^{-2}P_KDD^T)^{-1} P_K$. Plugging \eqref{eq:matrixinverse} into \eqref{eq:Si2}, we have
\begin{align}\label{eq:SiFinal}
    S_K = (I_n - \gamma^{-2}P^*DD^T)^{-1} (P_K - P^*)D(\gamma^2I_q - D^TP_KD)^{-1}D^T(P_K - P^*) (I_n - \gamma^{-2}DD^TP^*)^{-1}.
\end{align}
Therefore, \eqref{eq:Ui-U*Lemma} is obtained by \eqref{eq:UoptLemma} and \eqref{eq:SiFinal}.
\end{proof}

\begin{lemma}\label{lm:continuousPK}
    For any $K \in \mathcal{W}$, $P_K$ is continuous with respect to $K$, where $P_K$ is the unique positive-definite solution to \eqref{eq:AREforKPrelimi}.
\end{lemma}
\begin{proof}
The lemma follows from the implicit function theorem. Consider the function
\begin{align}
\begin{split}
    &F(K,P_K) := (A-BK)^TU_K(A-BK) - P_K + Q + K^TRK.
\end{split}
\end{align}    
Let $\mathcal{F}(\vect(K),\vect(P_K)) := \vect(F(K,P_K))$. Then,
\begin{align}\label{eq:partialF}
\begin{split}
    &\mathcal{F}(\vect(K),\vect(P_K))= [(A-BK)^T \otimes (A-BK)^T] \vect(U_K) \\
    &- \vect(P_K) + \vect(Q + K^TRK). 
\end{split}
\end{align}  
Noticing $U_K = (I_n - \gamma^{-2}P_KDD^T)^{-1}P_K$ and \cite[Theroem 9]{Magnus1985}, we have
\begin{align}\label{eq:UKderivative}
    &\frac{\partial \vect(U_K)}{\partial \vect(P_K)} = (P_K \otimes I_n) \frac{\partial \vect((I_n - \gamma^{-2}P_KDD^T)^{-1})}{\partial \vect(P_K)} \nonumber\\
    &+ [I_n \otimes (I_n - \gamma^{-2}P_KDD^T)^{-1}] \frac{\partial \vect(P_K)}{ \partial \vect(P_K)} \\ 
    &= \left\{[\gamma^{-2}P_KDD^T(I_n-\gamma^{-2}P_KDD^T)^{-1}] \right. \nonumber\\
    &\quad \left. \otimes (I-\gamma^{-2}P_KDD^T)^{-1}\right\}  + [I_n \otimes (I_n - \gamma^{-2}P_KDD^T)^{-1}] \nonumber\\
    &= (I_n - \gamma^{-2}P_KDD^T)^{-1} \otimes (I_n - \gamma^{-2}P_KDD^T)^{-1}, \nonumber
\end{align}
where the second equality follows from \cite[Equation (B.10)]{zhangarxiv2019}. Considering
\begin{align}\label{eq:A-BKLK*}
    &L_{K,*} = (\gamma^2I_q - D^TP_KD)^{-1}D^TP_K(A-BK) \\
    &A-BK + DL_{K,*} = (I_n - \gamma^{-2}DD^TP_K)^{-1}(A-BK),    \nonumber
\end{align}
and plugging  \eqref{eq:UKderivative} and \eqref{eq:A-BKLK*} into the derivative of \eqref{eq:partialF} yield
\begin{align}
\begin{split}
    &\frac{\partial \mathcal{F}(\vect(K),\vect(P_K))}{\partial \vect(P_K)} =  [(A-BK+DL_{K,*})^T \\
    &\otimes (A-BK+DL_{K,*})^T] - I_{n^2}.
\end{split}
\end{align}
Since $K \in \mathcal{W}$, by Lemma \ref{lm:realbounded}, $(A-BK+DL_{K,*})$ is stable. Then, $\frac{\partial \mathcal{F}(\vect(K),\vect(P_K))}{\partial \vect(P_K)}$ is invertible since $\sgmax(A-BK+DL_{K,*}) < 1$. By the implicit function theorem, $P_K$ is continuous with respect to $K$ for any $K \in \mathcal{W}$.
\end{proof}

\begin{lemma}\label{lm:E=0P=0}
Let $K \in \mathcal{W}$, $P_K$ be the positive-definite solution of \eqref{eq:AREforK}, and $K' = (R+B^TU_KB)^{-1}B^TU_KA$. Then, $(K - K')^T R (K - K') = 0$ is equivalent to $K = K^*$.
\end{lemma}
\begin{proof}
($\Rightarrow$): Since $R \succ 0$,  $(K - K')^T R (K - K') = 0$ implies $K = K'$. From \eqref{eq:AREforK}, we have
\begin{align}
    &(A-BK)^T U_K  (A-BK) - P_K + Q + K^TRK = 0 \nonumber \\
    &K = (R + B^TUB)^{-1}B^TU_KA.
\end{align}
Hence, $P_K \succ 0$ satisfies the GARE \eqref{eq:GARE}. Due to the uniqueness of the positive-definite solution to \eqref{eq:GARE}, it is deduced that $P_K=P^*$ and $K = K^*$.

($\Leftarrow$): Since $K = K^*$ and $P_K = P^*$, it follows that $K' = (R+B^TU_KB)^{-1}B^TU_KA = K$. Hence, $K = K'$ and $(K - K')^TR(K-K')=0$. 
\end{proof}

\begin{lemma}\label{lm:compactGh}
   For any $h \in \mathcal{H}$, $\mathcal{G}_h := \{ K \in \mathcal{W}| \Tr(P_K) \leq \Tr(P^*) +h \}$ is compact.
\end{lemma}
\begin{proof}
    Firstly, we prove the boundedness of $\mathcal{G}_h$. For any $K \in \mathcal{G}_h$, $P_K$ is bounded. Since $U_K \succeq 0$, it is seen from \eqref{eq:AREforK} that $ K^TR K \preceq P_K$. As $R \succ 0$, $K$ is bounded.

    Next, we prove the closeness of $\mathcal{G}_h$. Let $\{ K_{s}\}_{s=1}^{\infty}$ denote an arbitrary convergent sequence within $\mathcal{G}_h$ and $\lim_{s \to \infty}\norm{K_s - K_{lim}}_F = 0$. Since $P_K$ is continuous with respect to $K$ (Lemma \ref{lm:continuousPK}), we have $\lim_{s \to \infty}P_{K_s} = P_{K_{lim}}$. Since $\Tr(P_{K_{s}}) \leq \Tr(P^*) + h$, it follows that $\Tr(P_{K_{lim}}) \leq \Tr(P^*) + h$. In addition, $P_{K_{lim}} \succeq 0$ is obtained by Lemma \ref{lm:realbounded} and $P_{K_s} \succeq 0$. Since the pair ($K_{lim}, P_{K_{lim}}$) satisfies \eqref{eq:AREforK}, $K_{lim} \in \mathcal{W}$ is obtained by Lemma \ref{lm:realbounded}. As a result, $K_{lim} \in \mathcal{G}_h$ and $\mathcal{G}_h$ is closed. In summary, the compactness of $\mathcal{G}_h$ is demonstrated by Heine–Borel theorem.
\end{proof}

\begin{lemma}\label{lm:EiboundFi}
For any $h \in \mathcal{H}$ and $K \in \mathcal{G}_h$, let $K' := (R+B^TU_KB)^{-1}B^TU_KA$, and $E_{K} := (K' - K)^T(R+B^TU_KB)(K' - K)$. Then, there exists $a(h)>0$, such that 
\begin{align}
\begin{split}
    \norm{P_K - P^*}_F \leq a(h)\norm{E_K}_F. 
\end{split}
\end{align}
\end{lemma}
\begin{proof}
It follows from \eqref{eq:AREforKPrelimi} that
\begin{align} \label{eq:outerloopIte_Rewrite}
\begin{split}
    &(A-BK^*)^TU_K(A-BK^*) - P_K + Q + K^TRK \\
    &+ A^TU_KB(K^* - K) + (K^* - K)^TB^TU_KA \\
    &+ K^TB^TU_KBK - (K^*)^TB^TU_KBK^* = 0.
\end{split}
\end{align}
Subtracting \eqref{eq:Popt} from \eqref{eq:outerloopIte_Rewrite}, and considering $B^TU_KA = (R+B^TU_KB)K'$, we have
\begin{align}\label{eq:AoptDiffUiUopt1}
    &(A-BK^*)^T(U_K - U^*)(A-BK^*) - (P_K - P^*) \nonumber\\
    &+ K^T(R+B^TU_KB)K +(K')^T(R+B^TU_KB)(K^* - K) \nonumber\\
    & + (K^* - K)^T(R+B^TU_KB)K' \nonumber\\
    &- (K^*)^T(R+B^TU_KB)K^* = 0.
\end{align}
Completing the squares in \eqref{eq:AoptDiffUiUopt1} yields
\begin{align}\label{eq:AoptDiffUiUopt2}
    &(A-BK^*)^T(U_K - U^*)(A-BK^*) - (P_K - P^*) + E_K \nonumber\\
    &-  (K'-K^*)^T(R+B^TU_KB)(K' - K^*) = 0.
\end{align}
It follows from $(I_n - \gamma^{-2}DD^TP^*)^{-1}(A-BK^*) = (A-BK^*+DL^*) = A^*$ and Lemma \ref{lm:UiUoptQuad} that
\begin{align}\label{eq:AoptDiffUiUopt3}
    &(A^*)^T\Delta P_K(A^*) - \Delta P_K + E_K \nonumber\\
    &-  (K'-K^*)^T(R+B^TU_KB)(K' - K^*)  \\
    &+ (A^*)^T\Delta P_KD(\gamma^2I_q - D^TP_KD)^{-1}D^T \Delta P_K A^*  = 0. \nonumber
\end{align}
To simplify the notation, let $\Delta P_K := P_K - P^*$. According to Lemma \ref{lm:lyapunov}, we have 
\begin{align}\label{eq:AoptDiffUiUopt4}
\begin{split}
    \Delta P_K &\preceq \sum_{k=0}^\infty (A^*)^{T,k} \left[ E_K + (A^*)^T \Delta P_K D \right. \\
    &\left. (\gamma^2 I_q - D^T P_K D)^{-1}D^T \Delta P_K A^* \right](A^*)^{k}.
\end{split}
\end{align}
Taking the trace of \eqref{eq:AoptDiffUiUopt4} and using the cyclic property of trace and trace inequality in \cite[Lemma 1]{Wang1986} yields
\begin{align}\label{eq:AoptDiffUiUopt5}
    &\Tr(\Delta P_K) \le \Tr \left[ \sum_{k=0}^\infty (A^*)^{k}(A^*)^{T,k}E_K\right] + \Tr \left[ \sum_{k=1}^\infty (A^*)^{k} \right. \nonumber\\
    &\left. (A^*)^{T,k} \Delta P_K D(\gamma^2 I_q - D^T P_K D)^{-1}D^T \Delta P_K \right] \nonumber\\
    &\le a_1 \Tr(E_K) + a_2\gamma^{-2}\Tr[\Delta P_K \nonumber\\
    & ( I_n - \gamma^{-2}DD^T P_K )^{-1}DD^T \Delta P_K],
\end{align}
where 
\begin{align}
\begin{split}
    a_1 &= \norm{ \sum_{k=0}^\infty (A^*)^{k}(A^*)^{T,k}} , \\
    a_2 &= \norm{ \sum_{k=1}^\infty (A^*)^{k}(A^*)^{T,k}}.    
\end{split}
\end{align}

By Neumann series, for any $X \in \mathbb{R}^{n \times n}$ with $\norm{X} < 1$, it holds
\begin{align}
    \norm{(I-X)^{-1}} = \norm{\sum_{k=0}^\infty X^k} \le \frac{1}{1-\norm{X}}.
\end{align}
Therefore, for small enough $\Delta P_K$, we have
\begin{align}\label{eq:I-DDP}
    &\norm{( I_n - \gamma^{-2}DD^T P_K )^{-1}} \nonumber\\
    &=\norm{( I_n - \gamma^{-2}DD^T P^* -  \gamma^{-2}DD^T \Delta P_K)^{-1}} \\
    &\le \frac{\norm{(I_n - \gamma^{-2}DD^T P^*)^{-1}}}{1 - a_3\norm{\Delta P_K}}. \nonumber
\end{align}
where
\begin{align}
    a_3 = \gamma^{-2}\norm{(I_n - \gamma^{-2}DD^T P^*)^{-1}}\norm{DD^T}.
\end{align}
Using the trace inequality in \cite[Lemma 1]{Wang1986}, it follows from \eqref{eq:AoptDiffUiUopt5} and \eqref{eq:I-DDP} that
\begin{align}
    &\left(1-\frac{a_2a_3\norm{\Delta P_K}}{1 - a_3\norm{\Delta P_K}}\right)\Tr(\Delta P_K) \le a_1 \Tr(E_K).
\end{align}
Therefore, if 
\begin{align}
    \norm{\Delta P_K} \le \frac{1}{a_3 + 2 a_2 a_3} =: h_1,
\end{align}
it follows from Lemma \ref{lm:normInequ} that
\begin{align}
    \norm{\Delta P_K}_F \le \Tr(\Delta P_K) \le 2a_1 \Tr(E_K) \le 2a_1 \sqrt{n}\norm{E_K}_F .
\end{align}

When $ \norm{\Delta P_K} \ge h_1$, if follows from Lemma \ref{lm:E=0P=0} that $\norm{E_K}_F \neq 0$. Since $E_K$ is continuous with respect to $K$ (Lemma \ref{lm:continuousPK}) and the set $\mathcal{G}_h \cap \{K \in \mathcal{W}| \norm{\Delta P_K} \ge h_1 \}$ is compact (Lemma \ref{lm:compactGh}), there exists $a_4(h) > 0$, such that $\norm{E_K}_F \geq a_4(h)$. Hence, $\norm{\Delta P_K}_F \leq \Tr(\Delta P_K) \le \frac{h}{a_4(h)} \norm{E_K}_F$. By taking $a(h) = \max(2a_1 \sqrt{n},\frac{h}{a_4(h)} )$, we obtain that $\norm{\Delta P_K}_F \leq a(h) \norm{E_K}_F$.

\end{proof}

Now, we are ready to prove Theorem \ref{thm:outerloopGloballinear}.

\textbf{Proof of Theorem \ref{thm:outerloopGloballinear}}.  We can rewrite \eqref{eq:outloop_evalu} as
\begin{align}
\begin{split}
    &A_{i+1}^T U_i A_{i+1} + K_{i}^TB^TU_iBK_{i} - K_{i+1}^TB^TU_iBK_{i+1}  \\
    & + (K_{i+1} - K_{i})^TB^TU_iA + A^TU_iB(K_{i+1} - K_{i}) \\
    &- P_i + Q + K_i^TRK_i = 0.    
\end{split}
\end{align}
Since $(R+B^TU_iB)K_{i+1} = B^TU_iA$ from \eqref{eq:outerloop_update}, by completing the squares, we have
\begin{align}\label{eq:AnextPi}
\begin{split}
    &A_{i+1}^T U_i A_{i+1} - P_i + Q + K_{i+1}^TRK_{i+1}  + E_i  = 0,   
\end{split}
\end{align}
where $E_i = E_{K_i} = (K_{i+1} - K_{i})^T(R+B^TU_iB)(K_{i+1} - K_{i})$. Writing out \eqref{eq:outloop_evalu} for the $(i+1)$th iteration, subtracting it from \eqref{eq:AnextPi}, we can obtain that
\begin{align}\label{eq:(Ui-Ui+1)Ai+1}
\begin{split}
    &A_{i+1}^T (U_i-U_{i+1}) A_{i+1} - (P_i-P_{i+1}) + E_i  = 0.   
\end{split}      
\end{align}
From \eqref{eq:Ui}, the expression of $U_{i+1}$ is derived as
\begin{align}\label{eq:Ui-Ui+1}
\begin{split}
    U_{i+1} & = [I_n + P_{i+1}D(\gamma^2I_q - D^TP_{i+1}D)^{-1}D^T]P_{i+1} \\
    &= [I_n + P_{i+1}DD^T(\gamma^2I_n - P_{i+1}DD^T)^{-1}]P_{i+1} \\
    &= (I_n - \gamma^{-2}P_{i+1}DD^T)^{-1}P_{i+1} \\
    &= (I_n - \gamma^{-2}P_{i+1}DD^T)^{-1}P_{i+1}(I_n - \gamma^{-2}DD^TP_{i+1})(I_n - \gamma^{-2}DD^TP_{i+1})^{-1} \\
    &= (I_n - \gamma^{-2}P_{i+1}DD^T)^{-1}(P_{i} - \gamma^{-2}P_{i+1}DD^TP_{i+1})(I_n - \gamma^{-2}DD^TP_{i+1})^{-1} \\
    & \quad - (I_n - \gamma^{-2}P_{i+1}DD^T)^{-1}(P_{i} - P_{i+1})(I_n - \gamma^{-2}DD^TP_{i+1})^{-1}\\
    &\preceq U_i - (I_n - \gamma^{-2}P_{i+1}DD^T)^{-1}(P_{i} - P_{i+1})(I_n - \gamma^{-2}DD^TP_{i+1})^{-1},
\end{split}
\end{align}
where the last inequality is derived using \cite[Lemma B.1]{zhangarxiv2019}. Combining \eqref{eq:(Ui-Ui+1)Ai+1} and \eqref{eq:Ui-Ui+1}, we have
\begin{align} \label{eq:Ai+1Pdiff}
\begin{split}
    &A_{i+1}^T (I_n - \gamma^{-2}P_{i+1}DD^T)^{-1}(P_{i} - P_{i+1}) \\
    &\quad (I_n - \gamma^{-2}DD^TP_{i+1})^{-1} A_{i+1} - (P_i-P_{i+1})  + E_i \preceq 0.
\end{split}
\end{align}
Considering the expression of $L_{i+1,*}$ in \eqref{eq:LoptForK}, we have
\begin{align}\label{eq:Ai+1andOpt}
    &(I_n - \gamma^{-2}DD^TP_{i+1})^{-1} A_{i+1} \nonumber\\
    &= \left[ I_n + \gamma^{-2}DD^TP_{i+1}(I_n - \gamma^{-2}DD^TP_{i+1})^{-1} \right] A_{i+1} = A_{i+1,*},
\end{align}
As a consequence, \eqref{eq:Ai+1Pdiff} can be rewritten as
\begin{align}\label{eq:Ai+1optPdiff}
\begin{split}
    &A_{i+1,*}^T(P_{i} - P_{i+1})A_{i+1,*} - (P_i-P_{i+1}) + E_i \preceq 0.
\end{split}
\end{align}

By Lemma \ref{lm:outerloop_converge}, it follows that $\{P_{i}\}$ is monotonically decreasing and $\Tr(P_i) \leq \Tr(P_1)$ for any $i \in \mathbb{Z}_+$. Hence, given $K_1 \in \mathcal{G}_h$, $K_i \in \mathcal{G}_h$ for any $i \in \mathbb{Z}_+$. Following \eqref{eq:Ai+1optPdiff} and Lemma \ref{lm:lyapunov}, we have
\begin{align}\label{eq:DiffPiLowerBound1}
    (P_i-P_{i+1}) \succeq \sum_{t=0}^{\infty} (A_{i+1,*}^T)^tE_iA_{i+1,*}^t
\end{align}
Subtracting $P^*$ from both sides of \eqref{eq:DiffPiLowerBound1} and taking trace of \eqref{eq:DiffPiLowerBound1}, we have
\begin{align}\label{eq:DiffPiLowerBound2}
\begin{split}
    &\Tr(P_{i+1}-P^*) \le \Tr(P_{i}-P^*) - \Tr(E_i) \\
    &\le \Tr(P_{i}-P^*) - \norm{E_i}_F,
\end{split}
\end{align}
where the last inequality comes from Lemma \ref{lm:normInequ}. Considering Lemmas \ref{lm:normInequ} and \ref{lm:EiboundFi}, \eqref{eq:DiffPiLowerBound2} can be further derived as
\begin{align}
    \Tr(P_{i+1}-P^*) \leq \left(1-\frac{1}{\sqrt{n}a(h)}\right)\Tr(P_{i}-P^*).
\end{align}
The theorem is thus proved by setting $\alpha(h) = 1-\frac{1}{\sqrt{n}a(h)}$. From Lemma \ref{lm:outerloop_converge}, $P^* \preceq P_{i+1} \preceq P_{i}$, and thus $0 \le \Tr(P_{i+1}-P^*) \le \Tr(P_{i}-P^*)$. As a result, $\alpha(h) \in [0, 1)$.

\stepcounter{appidx}
\setcounter{equation}{0}

\section*{Appendix \Alph{appidx}: Proof of Theorem \ref{thm:innerloop_globallinear}}

Given an admissible feedback $K \in \mathcal{W}$, and starting from $L_{K,1}=0$, the inner-loop iteration is 
\begin{subequations}
\begin{align}
    &A_{K,j}^TP_{K,j}A_{K,j} - P_{K,j} + Q_K - \gamma^{2}L_{K,j}^T L_{K,j} = 0 \label{eq:innerloopEvaluforK}\\
    &L_{K,j+1} = (\gamma^2 I_q - D^T P_{K,j} D)^{-1}D^TP_{K,j}A_K \label{eq:innerloopUpdateforK}
\end{align}
\end{subequations}
Recall that $Q_K = Q + K^TRK$, $A_K = A-BK$ and $A_{K,j}=A-BK+L_{K,j}$. The following lemma states the monotonic convergence of the inner-loop iteration. 
\begin{lemma}\label{lm:innerloop convergence}
Suppose that the inner loop starts from the initial condition $L_{K,1}=0$. For any $K \in \mathcal{W}$, and $j \in \mathbb{Z}_+$, the following statements hold
\begin{enumerate}
    \item $A_{K,j}:=A-BK+L_{K,j}$ is stable;
    \item $P_K \succeq \cdots \succeq P_{K,j+1} \succeq P_{K,j} \succeq \cdots \succeq P_{K,1}$;
    \item $\lim_{j \to \infty} \norm{P_{K,j} - P_K}_F = 0$ and $\lim_{j \to \infty} \norm{L_{K,j} - L_{K,*}}_F = 0$.
\end{enumerate}
\end{lemma}
\begin{proof}
Considering the equalities $U_K = (I_n - \gamma^{-2}P_KDD^T)^{-1}P_K$, $A_{K,*} = (I_n - \gamma^{-2}DD^TP_K)^{-1}A_K$ from \eqref{eq:Ai+1andOpt}, and $L_{K,*} = (\gamma^2I_q - D^TP_KD)^{-1}D^TP_K(A-BK)$ from \eqref{eq:LoptForK}, we can rewrite \eqref{eq:AREforK} as
\begin{align}\label{eq:outloop_evalu_rewrite}
    A_{K,*}^TP_{K}A_{K,*} - P_{K} + Q_K - \gamma^{2}L_{K,*}^TL_{K,*} = 0.
\end{align}
Since $P_K \succ 0$ and $\norm{A_{K,*}} < 1$, $P_K - A_{K,*}^TP_KA_{K,*} \succ 0$. Therefore, from \eqref{eq:outloop_evalu_rewrite}, we have
\begin{align}\label{eq:Q_ibound}
     Q_K - \gamma^2 L_{K,*}^TL_{K,*} \succ 0.
\end{align}
Considering the equality $(\gamma^{2}I_q - D^TP_{K}D)L_{K,*} = D^TP_{K}A_K$ in \eqref{eq:LoptForK} and completing the squares in $(L_{K,*} - L_{K,j})^T(\gamma^2I_q - D^TP_KD)(L_{K,*} - L_{K,j})$, we can rewrite \eqref{eq:outloop_evalu_rewrite} as
\begin{align}\label{eq:outloop_evalu_rewrite2}
\begin{split}
    &A_{K,j}^TP_{K}A_{K,j} - P_{K} + Q_K - \gamma^2 L_{K,j}^TL_{K,j} +  (L_{K,*} - L_{K,j})^T(\gamma^2I_q - D^TP_KD)(L_{K,*} - L_{K,j}) = 0.
\end{split}
\end{align}
Subtracting \eqref{eq:innerloopEvaluforK} from \eqref{eq:outloop_evalu_rewrite2} and completing the squares yield
\begin{align} \label{eq:AijPdiff}
    A_{K,j}^T(P_{K} - P_{K,j})A_{K,j} - (P_{K} - P_{K,j}) + (L_{K,*} - L_{K,j})^T(\gamma^2I_q - D^TP_KD)(L_{K,*} - L_{K,j}) = 0.
\end{align}

We prove the first statement by induction. When $j=1$, the inner loop starts from $L_{K,j} = 0$. Since $K \in \mathcal{W}$, $A_{K,1} = A-BK_K +DL_{K,1} = A-BK_K$ is stable. Now, assume $A_{K,j}$ is stable for some $j \geq 1$. Since $A_{K,j}$ is stable and $(L_{K,*} - L_{K,j})^T(\gamma^2I_q - D^TP_KD)(L_{K,*} - L_{K,j}) \succeq 0$, Lemma \ref{lm:lyapunov} and \eqref{eq:AijPdiff} result in $P_K \succeq P_{K,j} \succ 0$. In addition, following the derivation of \eqref{eq:outloop_evalu_rewrite2},  \eqref{eq:outloop_evalu_rewrite} can be rewritten as
\begin{align}\label{eq:Aij+1Pi}
\begin{split}
    &A_{K,j+1}^TP_KA_{K,j+1} - P_K + Q_{K} - \gamma^2 L_{K,j+1}^TL_{K,j+1} \\
    &+ (L_{K,*} - L_{K,j+1})^T(\gamma^2I_q - D^TP_KD)(L_{K,*} - L_{K,j+1}) = 0.    
\end{split}
\end{align}
As $P_K \succeq P_{K,j}$, $(\gamma^2 I_q - D^TP_{K}D)^{-1} \succeq (\gamma^2 I_q - D^TP_{K,j}D)^{-1} $. As a consequence, $L_{K,j}^TL_{K,j} \preceq L_{K,*}^TL_{K,*}$ is obtained by comparing \eqref{eq:LoptForK} with \eqref{eq:innerloopUpdateforK}. Then, $Q_K - \gamma^2L_{K,j+1}^TL_{K,j+1} \succ 0$ follows from \eqref{eq:Q_ibound}. From \eqref{eq:Aij+1Pi} and \cite[Theorem 8.4]{book_Hespanha}, we see that $A_{K,j+1}$ is stable. \textit{A fortiori}, the first statement holds.

For the $(j+1)$th iteration, the policy evaluation step in \eqref{eq:innerloopEvaluforK} is 
\begin{align}\label{eq:Inneri+1}
    A_{K,j+1}^TP_{K,j+1} A_{K,j+1} - P_{K,j+1} + Q_K - \gamma^{2}L_{K,j+1}^TL_{K,j+1} = 0.
\end{align}
Subtracting \eqref{eq:innerloopEvaluforK} from \eqref{eq:Inneri+1}, considering $(\gamma^2I_q - D^TP_{K,j}D)L_{K,j+1} = D^TP_{K,j}A_K$ in \eqref{eq:innerloop_update}, and completing the squares, we have
\begin{align}\label{eq:Pijdiff}
\begin{split}
    &A_{K,j+1}^T(P_{K,j+1} - P_{K,j})A_{K,j+1} - (P_{K,j+1} - P_{K,j}) \\
    &+  (L_{K,j+1} - L_{K,j})^T(\gamma^2I_q - D^TP_{K,j}D)(L_{K,j+1} - L_{K,j}) = 0.
\end{split}
\end{align}
As $K \in \mathcal{W}$, by Lemma \ref{lm:realbounded}, $\gamma^2I_q - D^TP_{K}D \succ 0$. By the fact that $P_K \succeq P_{K,j}$, we have $\gamma^2I_q - D^TP_{K,j}D \succ 0$. Since $A_{K,j+1}$ is stable, by Lemma \ref{lm:lyapunov}, $P_{K,j+1} - P_{K,j} \succeq 0$. Hence, the second statement holds.

From the second state, the sequence $\{ P_{K,j}\}_{j=1}^{\infty}$ is monotonically increasing and bounded by $P_K$. Therefore, there exists a constant matrix $P_{K,\infty} \succeq 0$ such that $\lim_{j \to \infty} P_{K,j} = P_{K,\infty}$. It can be verified that $ P_{K,\infty}$ satisfies \eqref{eq:AREforK}. Due to the uniqueness of the positive-definite solution of \eqref{eq:AREforK}, we have $P_{K,\infty} = P_K$ and the third statement holds.
\end{proof}

\begin{lemma}\label{eq:Ebound_inn}
Given $K \in \mathcal{W}$, let $L$ be admissible, i.e. $A-BK+DL$ is stable, and recall from \eqref{eq:GNInner} that $L' = (\gamma^2I_q - D^TP_{K,L}D)^{-1}D^TP_{K,L}A_K$, where $P_{K,L}$ is defined in \eqref{eq:LyaforL} .  Define $E_{K,L} := (L - L')^T(\gamma^2I_q - D^TP_{K,L}D)(L - L') $. Then, there exists a constant $b(K) > 0$, such that 
\begin{align}
    \Tr(P_K - P_{K,L}) \leq b(K)\norm{E_{K,L}},
\end{align}
and 
\begin{align}
    b(K) := \Tr \left[\sum_{t=0}^{\infty}(A_{K,*})^t (A_{K,*}^T)^t\right].
\end{align}
\end{lemma}
\begin{proof}
Considering $(\gamma^2I_q - D^TP_{K,L}D)L' = D^TP_{K,L}A_K$ in \eqref{eq:GNInner} and completing the squares, \eqref{eq:LyaforL} can be rewritten as
\begin{align}
\begin{split}
    &A_{K,*}^TP_{K,L}A_{K,*} - P_{K,L} + Q_K - (L - L'
    )^T(\gamma^2I_q - D^TP_{K,L}D)(L - L') \\
    &+ (L_{K,*} - L')^T(\gamma^2I_q - D^TP_{K,L}D)(L_{K,*} - L') - \gamma^{2}L_{K,*}^TL_{K,*} = 0
\end{split}
\end{align}
Subtracting it from \eqref{eq:AREforK} results in 
\begin{align}
\begin{split}
    &A_{K,*}^T(P_K - P_{K,L})A_{K,*} - (P_K - P_{K,L}) + E_{K,L} \\
    &- (L_{K,*} - L')^T(\gamma^2I_q - D^TP_{K,L}D)(L_{K,*} - L')  = 0.
\end{split}
\end{align}
As $A_{K,*}$ is stable, by Lemma \ref{lm:lyapunov} we have
\begin{align}
    P_K - P_{K,L} \preceq \sum_{t=0}^{\infty}(A_{K,*}^T)^t E_{K,L}(A_{K,*})^t.
\end{align}
The following inequality can be derived by the cyclic property of trace and \cite[Lemma 1]{Wang1986}
\begin{align}
    \Tr(P_K - P_{K,L}) \leq \Tr \left[E_{K,L}\sum_{t=0}^{\infty}(A_{K,*})^t (A_{K,*}^T)^t\right] \leq \norm{E_{K,L}} \underbrace{ \Tr \left[\sum_{t=0}^{\infty}(A_{K,*})^t (A_{K,*}^T)^t\right]}_{b(K)} .
\end{align}
\end{proof}
Now, we are ready to prove Theorem \ref{thm:innerloop_globallinear}.

\textbf{Proof of Theorem \ref{thm:innerloop_globallinear}}.  Let $E_{K,j} = E_{K,L_j} = (L_{K,j+1} - L_{K,j})^T(\gamma^2I_q - D^TP_{K,j}D)(L_{K,j+1} - L_{K,j})$. 
By \eqref{eq:Pijdiff} and Lemma \ref{lm:lyapunov}, we have
\begin{align}
    \Tr(P_{K,j+1} - P_{K,j}) = \Tr \left[\sum_{t=0}^{\infty}(A_{K,j+1}^T)^t E_{K,j} (A_{K,j+1})^t \right].
\end{align}
Consequently, 
\begin{align}\label{eq:innerlinear}
     &\Tr(P_{K} - P_{K,j+1}) \leq \Tr(P_{K} - P_{K,j}) - \Tr(E_{K,j}) \\
     & \leq \Tr(P_{K} - P_{K,j}) - \norm{E_{K,j}} \leq \underbrace{(1-\frac{1}{ b(K)})}_{\beta(K)}\Tr(P_{K} - P_{K,j}) \nonumber
\end{align}
where the last inequality comes from Lemma \ref{eq:Ebound_inn}. Since $P_K \succeq P_{K,j}$, $\Tr(P_{K} - P_{K,j+1}) \geq 0$. Hence, $\beta(K) \in [0,1)$.

\stepcounter{appidx}
\setcounter{equation}{0}

\section*{Appendix \Alph{appidx}: Proof of Theorem \ref{thm:uniformConvergence}}

Let $M_K := \sum_{t=0}^{\infty}(A_{K,*})^t (A_{K,*}^T)^t$, where $A_{K,*} = A-BK +DL_{K,*}$ and $L_{K,*} = (\gamma^2I_q - D^TP_KD)^{-1}D^TP_K(A-BK)$. Since $K \in \mathcal{W}$, $A_{K,*}$ is stable by Lemma \ref{lm:realbounded}. By Lemma \ref{lm:lyapunov}, $M_K$ is the unique solution to 
\begin{align}
    A_{K,*} M_K A_{K,*}^T - M_K + I_n = 0.
\end{align}
Since $P_K$ is continuous in $K \in \mathcal{W}$ (Lemma \ref{lm:continuousPK}), $A_{K,*}$ is continuous in $K \in \mathcal{W}$. Hence, $M_K$ is continuous in $K \in \mathcal{W}$, and $b(K) = \Tr(M_K)$ is continuous in $K$. In addition, the set $\mathcal{G}_h:=\{K \in \mathcal{W}| \Tr(P_K) \leq \Tr(P^*) +h\}$ is compact (Lemma \ref{lm:compactGh}). Therefore, the upperbound of $b(K) = \Tr(M_K)$ exists on $K \in \mathcal{G}_h$, that is $b(K) \leq \bar{b}(h)$ for any $K \in \mathcal{G}_h$. Consequently,  for any $K \in \mathcal{G}_h$, $\beta(K) \leq \bar{\beta}(h)$, which is defined as
\begin{align}
    \bar{\beta}(h):= 1-\frac{1}{\bar{b}(h)}.
\end{align}

Given $K \in \mathcal{G}_h$, following Lemma \ref{lm:normInequ} and Theorem \ref{thm:innerloop_globallinear}, we have
\begin{align}
    &\norm{P_K - P_{K,j}}_F \leq \bar{\beta}^{j-1}(h)\Tr(P_K - P_{K,1}) \nonumber\\
    &\leq \bar{\beta}^{j-1}(h)\Tr(P_K) \leq \bar{\beta}^{j-1}(h)(\Tr(P^*)+h).
\end{align} 
Therefore, for any $K \in \mathcal{G}_h$ and $\epsilon>0$, if $j \geq \bar{j}(h,\epsilon) = \log_{\bar{\beta}}^{\frac{\epsilon}{\Tr(P^*)+h}}+1$, $\norm{P_{K,j} - P_K}_F \leq \epsilon$. Noting that $\bar{j}(h,\epsilon)$ is independent of $K$, the uniform convergence of the dual-loop algorithm follows readily.

\stepcounter{appidx}
\setcounter{equation}{0}

\section*{Appendix \Alph{appidx}: Proof of Theorem \ref{thm:outerISS}}

Recall that $\mathcal{G}_h := \{K \in \mathcal{W}| \Tr(P_K) \leq \Tr(P^*) + h \}$. The following lemma ensures that for $K \in \mathcal{G}_h$ and small perturbation $\Delta K \in \mathbb{R}^{m \times n}$, the updated policy becomes $K'+\Delta K$ that still belongs to $\mathcal{G}_h$. In other words, $\mathcal{G}_h$ is an invariant set under small disturbance.
\begin{lemma}\label{lm:staywithinSet}
 Let $K \in \mathcal{G}_h$, $K' := (R+B^TU_KB)^{-1}B^TU_KA$, and $\hat{K}' := K' + \Delta K $. Then, there exists $d(h) > 0$, such that $\hat{K}' \in \mathcal{G}_h$ if $\norm{\Delta K}_F \le d(h)$. 
\end{lemma}
\begin{proof}

Since $K \in \mathcal{G}_h$, it follows from Lemma \ref{lm:outerloop_converge} that $K' \in \mathcal{W}$. Suppose that $\hat{K}' \in \mathcal{W}$. According to Lemma \ref{lm:realbounded}, there exists a unique solution $\hat{P}_{K'} = \hat{P}_{K'}^T \succ 0$ to
\begin{subequations}
\begin{align}
    &(A-B\hat{K}')^T \hat{U}_{K'}(A-B\hat{K}') - \hat{P}_{K'} \nonumber\\
    &+ Q + (\hat{K}' )^TR\hat{K}' = 0, \label{eq:UK'}\\
    &\hat{U}_{K'} = \hat{P}_{K'} + \hat{P}_{K'}D(\gamma^2I_q - D^T\hat{P}_{K'}D)^{-1}D^T\hat{P}_{K'}.
\end{align}
\end{subequations}
In addition, we can rewrite \eqref{eq:AREforK} as
\begin{align}\label{eq:UKatK'}
\begin{split}
    &(A-B\hat{K}')^TU_K(A-B\hat{K}') - P_K + Q \\
    &+ (\hat{K}' - K)^TB^TU_KA + A^TU_KB(\hat{K}' - K) \\
    &+ K^T(R+B^TU_KB)K - (\hat{K}')^TB^TU_KB\hat{K}' = 0.  
\end{split}
\end{align}
Noticing $(R+B^TU_KB)K'=B^TU_KA$ and $\hat{K}' = K' + \Delta K$, \eqref{eq:UKatK'} implies
\begin{align}\label{eq:UKatK'2}
    &(A-B\hat{K}')^TU_K(A-B\hat{K}') - P_K + Q \nonumber\\
    &+ ({K}' - K)^T(R+B^TU_KB)K'  \\
    &+ (K')^T(R+B^TU_KB)({K}' - K)  \nonumber\\
    & + \Delta K^T(R+B^TU_KB)K' + (K')^T(R+B^TU_KB)\Delta K   \nonumber\\
    & + K^T(R+B^TU_KB)K - (\hat{K}')^TB^TU_KB\hat{K}' = 0.  \nonumber
\end{align}
Subtracting \eqref{eq:UK'} from \eqref{eq:UKatK'2} and completing the squares yield
\begin{align}\label{eq:UKatK'3}
    &(A-B\hat{K}')^T(U_K - \hat{U}_{K'})(A-B\hat{K}') \\
    &- (P_K - \hat{P}_{K'}) + E_K - \Delta{K}^T(R + B^TU_KB)\Delta{K} = 0.  \nonumber
\end{align}
From  Lemma \ref{lm:EiboundFi}, $E_K = (K'-K)^T(R+B^TU_KB)(K'-K)$. Using \cite[Lemma B.1]{zhangarxiv2019} and following the derivation of \eqref{eq:Ai+1optPdiff}, we have
\begin{align}\label{eq:PKPK'}
\begin{split}
    &(A-B\hat{K}'+D{L}_{\hat{K}',*})^T(P_K - \hat{P}_{K'})(A-B\hat{K}'+D{L}_{\hat{K}',*})  \\
    &- (P_K - \hat{P}_{K'})+ E_K - \Delta{K}^T(R + B^TU_KB)\Delta{K} \preceq 0,  
\end{split}
\end{align}
where ${L}_{\hat{K}',*} = (\gamma^2I_q - D^T\hat{P}_{K'}D)^{-1}D^T\hat{P}_{K'}(A-B\hat{K}')$. Since $\hat{K}'\in \mathcal{W}$, by Lemma \ref{lm:realbounded}, $(A-B\hat{K}'+D{L}_{\hat{K}',*})$ is stable. Using Lemma \ref{lm:lyapunov}, we have
\begin{align}\label{eq:PKPk'trace}
    &\Tr(P_K - \hat{P}_{K'}) \geq \Tr\left\{ \sum_{t=0}^\infty\left\{ (A-B\hat{K}'+D{L}_{\hat{K}',*})^{T,t}[E_K \right. \right. \nonumber\\
    &\left.\left. -\Delta{K}^T(R + B^TU_KB)\Delta{K}](A-B\hat{K}'+D{L}_{\hat{K}',*})^t \right\} \right\} 
\end{align}
Let 
\begin{align}
\begin{split}
c(\hat{K}') =& \|\sum_{t=0}^\infty(A-B\hat{K}'+D{L}_{\hat{K}',*})^t \\ 
&(A-B\hat{K}'+D{L}_{\hat{K}',*})^{T,t}\|
\end{split}
\end{align}
and $d_1(h) = \sup_{K\in\mathcal{G}_h}\norm{R+B^TU_KB}$. Then, by Lemmas \ref{lm:normInequ} and \ref{lm:EiboundFi}, and \cite[Lemma 1]{Wang1986}, \eqref{eq:PKPk'trace} implies
\begin{align}\label{eq:PKPk'trace2}
\begin{split}
    &\Tr(\hat{P}_{K'} - P^*) \leq (1-\frac{1}{\sqrt{n}{a}(h)})\Tr({P}_{K} - P^*) \\
    &+c(\hat{K}') d_1(h)\norm{\Delta K}_F^2.
\end{split}
\end{align}
Therefore, if $\norm{\Delta K}_F^2 \leq \frac{h}{c(\hat{K}') d_1(h)\sqrt{n}{a}(h)}$, it is ensured that $\Tr(\hat{P}_{K'} - P^*) \leq h$, i.e. $\hat{K}' \in \mathcal{G}_h$. Let $\bar{c}(h) = \sup_{K \in \mathcal{G}_h }c(K)$. Since $P_K$ is continuous in $K$ (Lemma \ref{lm:continuousPK}) and $L_{K,*}$ defined in \eqref{eq:LoptForK} is continuous in $P_K$ and $K$,  $c(K)$ is continuous with respect to $K$ and $\bar{c}(h) < \infty$. Therefore, if 
\begin{align}
    \norm{\Delta K}_F \leq \left(\frac{h}{\bar{c}(h) d_1(h)\sqrt{n}{a}(h)}\right)^{\frac{1}{2}} =: d(h),
\end{align}
it is ensured that $\hat{K}' \in \mathcal{G}_h$. In other words, $\mathcal{B}(K',d(h)) = \{K \in \mathbb{R}^{m \times n}| \norm{K - K'}_F \le d(h)\} \subset \mathcal{G}_h$.

Next, we prove that $\hat{K}' \in \mathcal{W}$ by contradiction. If $\hat{K}' \notin \mathcal{W}$, it follows that $\hat{K}' \notin \mathcal{B}(K',d(h))$. Hence, $\norm{\Delta K} > d(h)$, which contradicts with the condition $\norm{\Delta K} < d(h)$.

\end{proof}

Now, we prove Theorem \ref{thm:outerISS} and Corollary \ref{cor:dual-loopaccu}.

\textbf{Proof of Theorem \ref{thm:outerISS}.}
From Lemma \ref{lm:staywithinSet} and given an initial admissible policy $\hat{K}_1 \in \mathcal{G}_h$, it is seen that if $\norm{\Delta K}_\infty \leq d(h)$, $\hat{K}_i < \mathcal{G}_{h}$ for any $i \in \mathbb{Z}_+$. In \eqref{eq:PKPk'trace2}, considering $\hat{P}_i$ and $\hat{P}_{i+1}$ as $P_K$ and $\hat{P}_{K'}$, respectively,  we have
\begin{align}
\begin{split}
    & \Tr(\hat{P}_{i+1} - P^*) \leq (1-\frac{1}{\sqrt{n}{a}(h)})\Tr(\hat{P}_{i} - P^*) \\
    &+ \bar{c}(h)d_1(h)\norm{\Delta K_{i+1}}_F^2.
\end{split}
\end{align}
Repeating the above inequality from $i=1$ yields
\begin{align}
    &\Tr(\hat{P}_{i+1} - P^*) \leq (1-\frac{1}{\sqrt{n}{a}(h)})^i\Tr(\hat{P}_{1} - P^*)  \nonumber\\
    &  + \sqrt{n}{a}(h) \bar{c}(h)d_1(h)\norm{\Delta K}_\infty^2 
\end{align}
From Lemma \ref{lm:normInequ}, it follows that 
\begin{align}
\begin{split}
    &\norm{\hat{P}_{i} - P^*}_F \leq (1-\frac{1}{\sqrt{n}{a}(h)})^{i-1}\sqrt{n} \norm{\hat{P}_{1} - P^*}_F \\
    &+ \sqrt{n}{a}(h) \bar{c}(h)d_1(h)\norm{\Delta K}_\infty^2.
\end{split}
\end{align}
Thus, $\kappa_1(\cdot,\cdot)$ defined by  $\kappa_1(\norm{\hat{P}_{1} - P^*}_F,i) := (1-\frac{1}{\sqrt{n}{a}(h)})^{i-1}\sqrt{n} \norm{\hat{P}_{1} - P^*}_F$ is a $\mathcal{KL}$-function, and $\xi_1(\cdot)$ defined by $\xi_1(\norm{\Delta K}_\infty) = \sqrt{n}{a}(h) \bar{c}(h)d_1(h)\norm{\Delta K}_\infty^2$ is a $\mathcal{K}$-function. Therefore, the inexact outer-loop iteration is small-disturbance ISS.

\textbf{Proof of Corollary \ref{cor:dual-loopaccu}.}
For each outer-loop iteration of Algorithm \ref{alg:IteAlg_model}, $P_{i,\Bar{j}}$ instead of $P_i$ is used to update the policy. This leads to the disturbance at each iteration
\begin{align}
\begin{split}
    \Delta K_{i+1} & = (R+B^TU_{i,\Bar{j}}B)^{-1}B^TU_{i,\Bar{j}}A \\
    &-(R+B^TU_iB)^{-1}B^TU_iA.
\end{split}
\end{align}
As $K' = (R+B^TU_KB)^{-1}B^TU_KA$ is continuously differentiable in $U_K$, and $U_K$ is continuously differentiable in $P_K$, $K'$ is Lipschitz continuous in $P_K \in \{P\succ 0| \Tr(P) \leq \Tr(P^*)+h\}$. Consequently, there exists $d_2(h)>0$, 
\begin{align}\label{eq:DeltaKLipchi}
     \norm{\Delta K_{i+1}}_F \le d_2(h) \norm{P_i - P_{i,\Bar{j}}}_F.
\end{align}

By Theorem \ref{thm:outerISS},  for any $\epsilon>0$, there exist $d_3(h,\epsilon)>0$ and $\bar{i} \in \mathbb{Z}_+$, such that, if $K_1 \in \mathcal{G}_h$ and $\norm{\Delta K}_\infty \leq d_3(h,\epsilon)$, $K_i \in \mathcal{G}_h$ for all $i \in \mathbb{Z}_+$ and 
\begin{align}\label{eq:PbariPopt}
    \norm{P_{\bar{i}} - P^*}_F \leq (1/2)\epsilon. 
\end{align}
By Theorem \ref{thm:uniformConvergence}, there exists $\bar{j}(h,\epsilon) \in \mathbb{Z}_+$, such that for any $i \in \mathbb{Z}_+$
\begin{align}\label{eq:PijclosetoPi}
    \norm{P_i - P_{i,\Bar{j}}}_F \leq \min[{d_3}/{d_2},(1/2)\epsilon].
\end{align}
Therefore, \eqref{eq:DeltaKLipchi} and \eqref{eq:PijclosetoPi} imply $\norm{\Delta K}_\infty \leq d_3(h,\epsilon)$. By norm's triangle inequality, \eqref{eq:PbariPopt} and \eqref{eq:PijclosetoPi}, we have
\begin{align}
\begin{split}
     &\norm{P_{\bar{i},\bar{j}} - P^*}_F\leq \\
     & \norm{P_{\bar{i},\bar{j}} -P_{\bar{i}}}_F + \norm{P_{\bar{i}} - P^*}_F \leq \epsilon.
\end{split}
\end{align}

\stepcounter{appidx}
\setcounter{equation}{0}

\section*{Appendix \Alph{appidx}: Proof of Theorem \ref{thm:innerISS}}

The following lemma ensures the stability of the closed-loop system with the feedback gain $\hat{L}_{i,j}$ generated from the inexact inner loop. 
\begin{lemma}\label{lm:innerloop_smalldistur}
Given $\hat{K}_i \in \mathcal{W}$, there exists a constant $e(\hat{K}_i)>0$, such that $\hat{A}_{i,j} = A-B\hat{K}_i + D\hat{L}_{i,j}$ is stable for all $j \in \mathbb{Z}_+$, as long as $\norm{\Delta L_i}_\infty < e(\hat{K}_i)$.
\end{lemma}
\begin{proof}
This lemma is proven by induction. To simplify the notation, the following variables are defined to denote the inner-loop update without disturbance
\begin{align}\label{eq:LUpdate}
\begin{split}
     &\tilde{L}_{i,j+1} := (\gamma^2 I_q - D^T\hat{P}_{i,j}D)^{-1}D^T\hat{P}_{i,j}\hat{A}_i, \\
    &\hat{A}_i = A - B \hat{K}_i.   
\end{split}
\end{align}
Then, $\hat{L}_{i,j+1} = \tilde{L}_{i,j+1} + \Delta L_{i,j+1}$. Since $\hat{K}_i \in \mathcal{W}$ and $\hat{L}_{i,1}=0$, $\hat{A}_{i,1} = A - B\hat{K}_i + D\hat{L}_{i,1}$ is stable by Lemma \ref{lm:realbounded}. By induction, assume that $\hat{A}_{i,j}=A - B\hat{K}_i + D\hat{L}_{i,j}$ is stable for some $j\in \mathbb{Z}_+$. Considering $\hat{U}_i = (I_n - \gamma^{-2}\hat{P}_iDD^T)^{-1}\hat{P}_i$, $\hat{L}_{i,*} = (\gamma^2I_q - D^T\hat{P}_iD)^{-1}D^T\hat{P}_i(A-B\hat{K}_i)$, and $(A-B\hat{K}_i+D\hat{L}_{i,*}) = (I_n - \gamma^{-2}DD^T\hat{P}_i)^{-1}(A-B\hat{K}_i)$, we can rewrite \eqref{eq:outloopInext_evalu} as
\begin{align}\label{eq:AREforKi1}
\begin{split}
     &(A-B\hat{K}_i+D\hat{L}_{i,*})^T\hat{P}_i(A-B\hat{K}_i+D\hat{L}_{i,*}) \\
    &- \hat{P}_i + \hat{Q}_i - \gamma^{2}\hat{L}_{i,*}^T\hat{L}_{i,*} = 0.   
\end{split}
\end{align}
Subtracting the $j$th iteration of \eqref{eq:innerloopInext_eval} from \eqref{eq:AREforKi1} and completing the squares, we have
\begin{align} \label{eq:hatAijPdiff}
\begin{split}
    &\hat{A}_{i,j}^T(\hat{P}_{i} - \hat{P}_{i,j})\hat{A}_{i,j} - (\hat{P}_{i} - \hat{P}_{i,j}) + \\
    &(\hat{L}_{i,*} - \hat{L}_{i,j})^T(\gamma^2I_q - D^T\hat{P}_iD)(\hat{L}_{i,*} - \hat{L}_{i,j}) = 0.    
\end{split}
\end{align}
Since $ \hat{A}_{i,j}$ is stable, from Lemma \ref{lm:lyapunov}, $\hat{P}_i \succeq \hat{P}_{i,j}$, where $\hat{P}_i$ is from \eqref{eq:outloopInext_evalu}. 
By completing the squares, \eqref{eq:AREforKi1} implies
\begin{align}\label{eq:hatAj+1}
    & 0=\hat{A}_{i,j+1}^T\hat{P}_{i}\hat{A}_{i,j+1} - \hat{P}_{i} + \hat{Q}_i - \gamma^{2}\tilde{L}_{i,j+1}^T\tilde{L}_{i,j+1}-\hat{\Omega}_{i,j+1}  \nonumber\\
    &+(\tilde{L}_{i,j+1} - \hat{L}_{i,*})^T(\gamma^2I_q - D^T\hat{P}_{i}D)(\tilde{L}_{i,j+1} - \hat{L}_{i,*}), 
\end{align}
where
\begin{align}
    & \hat{\Omega}_{i,j+1} =  \Delta L_{i,j+1}^T D^T\hat{P}_iD \tilde{L}_{i,j+1} + \tilde{L}_{i,j+1}^T D^T\hat{P}_iD \Delta L_{i,j+1} \nonumber\\
    &+ \Delta{L}_{i,j+1}^T D^T\hat{P}_iD \Delta L_{i,j+1} \\
    &+ \hat{A}_i^T \hat{P_i}D\Delta{L}_{i,j+1} + \Delta{L}_{i,j+1}^TD^T\hat{P_i}\hat{A}_i. \nonumber   
\end{align}
Since $\hat{P}_{i} \succeq \hat{P}_{i,j}$, it follows that $\hat{L}^T_{i,*}\hat{L}_{i,*} \succeq \tilde{L}^T_{i,j+1}\tilde{L}_{i,j+1}$ and $\norm{\hat{L}_{i,*}} \geq \norm{\tilde{L}_{i,j+1}}$. As a consequence,  
\begin{align}
    &\norm{\hat{\Omega}_{i,j+1}} \leq e_1(\hat{K}_i)\norm{\Delta L_{i,j+1}}   + e_2(\hat{K}_i)\norm{\Delta L_{i,j+1}}^2.    
\end{align}
where
\begin{align}
    e_1(\hat{K}_i) &= (2\norm{D^T\hat{P}_iD}\norm{\hat{L}_{i,*}} + 2\norm{D^T\hat{P}_i\hat{A}_i}), \nonumber\\
    e_2(\hat{K}_i) &= \norm{D^T\hat{P}_iD}. \nonumber
\end{align}
Following Lemma \ref{lm:realbounded}, we know that $\hat{P}_i \succ 0$ and $\hat{A}_{i,*} = A-B\hat{K}_i + D\hat{L}_{i,*}$ is stable. Therefore, by Lemma \ref{lm:lyapunov} and \eqref{eq:AREforKi1}, $ \hat{Q}_i - \gamma^2\hat{L}^T_{i,*}\hat{L}_{i,*} \succ 0$, and $e_3(\hat{K}_i) := \sgmin( \hat{Q}_i - \gamma^2\hat{L}^T_{i,*}\hat{L}_{i,*})>0$. Hence, if $\norm{\Delta L_{i,j+1}}$ satisfies
\begin{align}
    \norm{\Delta L_{i,j+1}} \leq \frac{-e_1+\sqrt{e_1^2 + 2e_2e_3}}{2e_2} := e(\hat{K}_i),
\end{align}
we have
\begin{align}\label{eq:Qibound}
    \hat{Q}_i - \gamma^2\hat{L}^T_{i,*}\hat{L}_{i,*} - \hat{\Omega}_{i,j+1} \succ \frac{1}{2}e_3(\hat{K}_i) I_n. 
\end{align}
As $\hat{L}^T_{i,*}\hat{L}_{i,*} \succeq \tilde{L}^T_{i,j+1}\tilde{L}_{i,j+1}$, $\hat{Q}_i - \gamma^2\tilde{L}^T_{i,j+1}\tilde{L}_{i,j+1} - \hat{\Omega}_{i,j+1} \succ 0$. $\hat{A}_{i,j+1}$ is stable as a result of \eqref{eq:hatAj+1} and \cite[Theorem 8.4]{book_Hespanha}. Therefore, for any $j \in \mathbb{Z}_+$, $\hat{A}_{i,j}$ is stable.

\end{proof}

\textbf{Proof of Theorem \ref{thm:innerISS}}. 
For the $j$th iteration, \eqref{eq:innerloopInext_eval} can be rewritten as
\begin{align}\label{eq:hatAj+1Pj}
     &\hat{A}_{i,j+1} ^T\hat{P}_{i,j}\hat{A}_{i,j+1} - \hat{P}_{i,j} + \hat{Q}_i  - \gamma^2\tilde{L}_{i,j+1}^T\tilde{L}_{i,j+1} \nonumber \\
     &- (\hat{L}_{i,j} - \tilde{L}_{i,j+1})^T(\gamma^2I_q - D^T\hat{P}_{i,j}D)(\hat{L}_{i,j} - \tilde{L}_{i,j+1}) \nonumber \\
     &  - \gamma^{2} \Delta L_{i,j+1}^T \tilde{L}_{i,j+1} - \gamma^{2} \tilde{L}_{i,j+1}^T\Delta L_{i,j+1} 
 \nonumber\\
     &- \Delta L_{i,j+1}^TD^T \hat{P}_{i,j}D \Delta{L}_{i,j+1} = 0.
\end{align}
Subtracting \eqref{eq:hatAj+1Pj} from the $(j+1)$th iteration of \eqref{eq:innerloopInext} yields
\begin{align}\label{eq:hatAj+1Pdiff}
     &\hat{A}_{i,j+1} ^T(\hat{P}_{i,j+1} - \hat{P}_{i,j})\hat{A}_{i,j+1} - (\hat{P}_{i,j+1} - \hat{P}_{i,j})  \nonumber\\
     &+ \underbrace{(\hat{L}_{i,j} - \tilde{L}_{i,j+1})^T(\gamma^2I_q - D^T\hat{P}_{i,j}D)(\hat{L}_{i,j} - \tilde{L}_{i,j+1})}_{\hat{E}_{i,j}}  \nonumber\\
     &  - \Delta L_{i,j+1}^T ( \gamma^2I_q - D^T \hat{P}_{i,j}D ) \Delta{L}_{i,j+1} = 0. 
\end{align}

When $\norm{\Delta L_{i}}_\infty < e(\hat{K}_i)$, by Lemma \ref{lm:innerloop_smalldistur}, $\hat{A}_{i,j+1}$ is stable. Following Lemma \ref{lm:lyapunov}, we have
\begin{align}
\begin{split}
     &\Tr(\hat{P}_{i,j+1} - \hat{P}_{i,j}) = \Tr \left\{\sum_{t=0}^{\infty} (\hat{A}_{i,j+1}^T)^t \left[\hat{E}_{i,j}  \right.\right. \\
    & \left.\left. - \Delta L_{i,j+1}^T( \gamma^2I_q - D^T \hat{P}_{i,j}D ) \Delta{L}_{i,j+1} \right](\hat{A}_{i,j+1})^t  \right\}.    
\end{split}
\end{align}
Consequently, by Lemma \ref{eq:Ebound_inn},
\begin{align}\label{eq:Phatdiffbound}
\begin{split}
    &\Tr(\hat{P}_i - \hat{P}_{i,j+1} ) \leq (1-\frac{1}{b(\hat{K}_i)})\Tr(\hat{P}_i - \hat{P}_{i,j}) \\
    &  +  \gamma^2 \norm{\Delta L_{i}}^2_\infty \Tr \left[\sum_{t=0}^{\infty} (\hat{A}_{i,j+1}^T)^t  (\hat{A}_{i,j+1})^t  \right].  
\end{split}
\end{align}

Let $\hat{M}_{i,j+1} := \sum_{t=0}^{\infty} (\hat{A}_{i,j+1}^T)^t  (\hat{A}_{i,j+1})^t$, and by Lemma \ref{lm:lyapunov}, $\hat{M}_{i,j+1}$ satisfies
\begin{align}\label{eq:Mij}
    \hat{A}_{i,j+1}^T\hat{M}_{i,j+1}\hat{A}_{i,j+1} - \hat{M}_{i,j+1} + I_n = 0.
\end{align}
Multiplying both sides of \eqref{eq:Mij} by $\frac{1}{2}e_3(\hat{K}_i)$ and subtracting it from \eqref{eq:hatAj+1}, we have
\begin{align}
     &\hat{A}_{i,j+1}^T(\hat{P}_{i} - \frac{1}{2}e_3 \hat{M}_{i,j+1} )\hat{A}_{i,j+1} - (\hat{P}_{i} - \frac{1}{2}e_3 \hat{M}_{i,j+1} )  \nonumber\\
     &+ \hat{Q}_i  - \gamma^{2}\tilde{L}_{i,j+1}^T\tilde{L}_{i,j+1}  - \hat{\Omega}_{i,j+1} - \frac{1}{2}e_3 I_n\\
     &+(\tilde{L}_{i,j+1} - \hat{L}_{i,*})^T(\gamma^2I_q - D^T\hat{P}_{i}D)(\tilde{L}_{i,j+1} - \hat{L}_{i,*}) = 0. \nonumber
\end{align}
By \eqref{eq:Qibound} and Lemma \ref{lm:lyapunov}, $\hat{P}_{i} - \frac{1}{2}e_3(\hat{K}_i) \hat{M}_{i,j+1} \succeq 0$. As a consequence $\Tr(\hat{M}_{i,j+1}) \leq 2/e_3\Tr(\hat{P}_i)$.

From \eqref{eq:Phatdiffbound}, we have
\begin{align}
\begin{split}
     &\Tr(\hat{P}_i - \hat{P}_{i,j+1} ) \leq  (1-\frac{1}{b(\hat{K}_i)})\Tr(\hat{P}_i - \hat{P}_{i,j}) \\
     &+ \frac{2}{e_3(\hat{K}_i)}\Tr(\hat{P}_i)\gamma^2\norm{\Delta L_{i}}_\infty^2.
\end{split}
\end{align}
Using Lemma \ref{lm:normInequ} and repeating the above argument for $j,j-1,\cdots,1$, it follows that
\begin{align}
\begin{split}
    &\norm{\hat{P}_i - \hat{P}_{i,j}}_F \leq(1-\frac{1}{b(\hat{K}_i)})^{j-1}\sqrt{n}\norm{\hat{P}_i - \hat{P}_{i,1}}_F \\
    &+ \frac{2}{e_3(\hat{K}_i)}b(\hat{K}_i)\Tr(\hat{P}_i)\gamma^2\norm{\Delta L_{i}}_\infty^2.     
\end{split}
\end{align}
Clearly, $\kappa_2(\cdot,\cdot)$ defined as $\kappa_2(\norm{\hat{P}_i - \hat{P}_{i,1}}_F,j) = (1-\frac{1}{b(\hat{K}_i)})^{j-1}\sqrt{n}\norm{\hat{P}_i - \hat{P}_{i,1}}_F$ is a $\mathcal{KL}$-function, and $\xi_2(\cdot)$ defined as $\xi_2(\norm{\Delta L_{i}}_\infty) = \frac{2}{e_3(\hat{K}_i)}b(\hat{K}_i)\Tr(\hat{P}_i)\gamma^2\norm{\Delta L_{i}}_\infty^2$ is a $\mathcal{K}$-function. Therefore, we can conclude that the inexact inner-loop iteration is ISS.

\stepcounter{appidx}
\setcounter{equation}{0}

\section*{Appendix \Alph{appidx}: Proof of Theorem \ref{thm:initialController}}

Since when $\epsilon=\mu=0$ and $\tau \to \infty$, the feasible set of the LMIs \eqref{eq:hinfLMISysID1} and \eqref{eq:hinfLMISysID2} is nonempty, by continuity, \eqref{eq:hinfLMISysID1} and \eqref{eq:hinfLMISysID2} has a solution for sufficiently small $\epsilon$ and $\mu$ and sufficiently large $\tau$.

Equation \eqref{eq:hinfLMISysID1} implies that
\begin{align} \label{eq:hinfLMISysID3}
\begin{split}
&\begin{bmatrix}
-{W} &* &* &* \\
0  &-\gamma ^2  I_q &* &* \\
{A}{W} - {B}{V} &{D} &-{W} &* \\
C{W} - E{V} &0 &0 &-I_p
\end{bmatrix} \\
&+    \begin{bmatrix}
\epsilon I_n &* &* &* \\
0  &\epsilon I_q &* &* \\
\tilde{A}_\tau{W} - \tilde{B}_\tau{V} &\tilde{D}_\tau &\epsilon I_n &* \\
0 &0 &0 &\epsilon I_p
\end{bmatrix} \prec 0,
\end{split}
\end{align}
where $\tilde{A}_\tau = \hat{A}_\tau - A$, $\tilde{B}_\tau = \hat{B}_\tau - B$, and $\tilde{D}_\tau = \hat{D}_\tau - D$. By Schur complement lemma, it follows from \eqref{eq:hinfLMISysID2} that
\begin{align}\label{eq:hinfLMISysID4}
    I_n - \mu^2[W, -V^T]\begin{bmatrix}
        W \\
        -V
    \end{bmatrix} \succ 0.
\end{align}

Since the following relation holds almost surely 
\begin{align}
    \lim_{\tau \to \infty }\tilde{A}_\tau = 0, \quad \lim_{\tau \to \infty }\tilde{B}_\tau = 0, \quad \lim_{\tau \to \infty }\tilde{D}_\tau = 0, 
\end{align}
there exists $\tau^*(\epsilon,\mu) > 0$, such that for all $\tau > \tau^*(\epsilon,\mu)$, the following inequality holds \textit{almost surely}:
\begin{align}
\norm{\tilde{A}_\tau}_F \le \frac{\epsilon\mu}{\sqrt{2}},\quad \norm{\tilde{B}_\tau}_F \le \frac{\epsilon\mu}{\sqrt{2}}, \quad \norm{\tilde{D}_\tau}_F \le \frac{\epsilon}{{2}}.
\end{align}    
Consequently, $\norm{[\tilde{A}, \tilde{B}]} \le \frac{\epsilon\mu}{\sqrt{2}}$, and
\begin{align}\label{eq:WVDWV}
    &[0, {W}\tilde{A}^T - {V}^T\tilde{B}^T]
    \begin{bmatrix}
        \epsilon I_q &\tilde{D}^T \\
        \tilde{D} & \epsilon I_n
    \end{bmatrix}^{-1}
    \begin{bmatrix}
        0 \\
        \tilde{A}{W} - \tilde{B}{V}
    \end{bmatrix} \nonumber\\
    &= [W, -V^T]
    \begin{bmatrix}
        \tilde{A}^T \\
        \tilde{B}^T
    \end{bmatrix}
    (\epsilon I_n -\frac{1}{\epsilon}\tilde{D}\tilde{D}^T)^{-1}
    \begin{bmatrix}
        \tilde{A} &\tilde{B}
    \end{bmatrix}
    \begin{bmatrix}
        W \\
        -V
    \end{bmatrix} \nonumber\\
    &\preceq 
    \frac{2}{\epsilon}[W, -V^T]
    \begin{bmatrix}
        \tilde{A}^T \\
        \tilde{B}^T
    \end{bmatrix}
    \begin{bmatrix}
        \tilde{A} &\tilde{B}
    \end{bmatrix}
    \begin{bmatrix}
        W \\
        -V
    \end{bmatrix} \nonumber \\
    &\preceq \epsilon \mu^2 [W, -V^T]\begin{bmatrix}
        W \\
        -V
    \end{bmatrix}.
\end{align}    
Therefore, by combining \eqref{eq:hinfLMISysID4} and \eqref{eq:WVDWV}, one can obtain
\begin{align}\label{eq:hinfLMISysID5}
    \epsilon I_n - [0, {W}\tilde{A}^T - {V}^T\tilde{B}^T]
    \begin{bmatrix}
        \epsilon I_q &\tilde{D}^T \\
        \tilde{D} & \epsilon I_n
    \end{bmatrix}^{-1}
    \begin{bmatrix}
        0 \\
        \tilde{A}{W} - \tilde{B}{V}
    \end{bmatrix} \succeq 0.
\end{align}
Using Schur complement lemma, we can obtain that 
the second term in \eqref{eq:hinfLMISysID3} is positive semi-definite, and therefore, the first term in \eqref{eq:hinfLMISysID3} is negative definite. By Lemma \ref{lm:realbounded}, it follows that $K = VW^{-1} \in \mathcal{W}$.
    

\bibliographystyle{ieeetr}        
\bibliography{IEEEabrv,mybibfile.bib}

\end{document}